\documentclass[preprint,12pt]{elsarticle}
\usepackage{multicol}

\usepackage{amsmath}
\usepackage{amsfonts}
\usepackage{amssymb}
\usepackage{amsthm}
\usepackage{verbatim}

\usepackage{graphicx}
\usepackage{epic}
\usepackage{eepic}
\usepackage{epsfig,float}
\usepackage{pdfsync}
\usepackage{bm}

\renewcommand{\le}{\leqslant}
\renewcommand{\ge}{\geqslant}

\newcommand{\ol}{\overline}
\newcommand{\eps}{\varepsilon}
\newcommand{\emp}{\emptyset}

\newcommand{\Sig}{\Sigma}

\newcommand{\bi}{\begin{itemize}}
\newcommand{\ei}{\end{itemize}}
\newcommand{\be}{\begin{enumerate}}
\newcommand{\ee}{\end{enumerate}}
\newcommand{\bd}{\begin{description}}
\newcommand{\ed}{\end{description}}
\newcommand{\eq}{\end{quote}}

\newcommand{\etc}{\mbox{\it etc.}}

\newcommand{\eg}{\mbox{\it e.g.}}

\newcommand{\cA}{{\mathcal A}}
\newcommand{\cB}{{\mathcal B}}

\newcommand{\cD}{{\mathcal D}}

\newcommand{\cN}{{\mathcal N}}

\newcommand{\cU}{{\mathcal U}}
\newcommand{\cX}{{\mathcal X}}

\newcommand{\rev}{R}
\newcommand{\deter}{D}
\newcommand{\mini}{M}
\newcommand{\trim}{T}

\newcommand{\qedb}{\hfill$\blacksquare$} 

\newtheorem{proposition}{Proposition}
\newtheorem{example}{Example}
\newtheorem{lemma}{Lemma}
\newtheorem{theorem}{Theorem}
\newtheorem{corollary}{Corollary}
\newtheorem{definition}{Definition}

\journal{Theoretical Computer Science}

\begin{document}

\begin{frontmatter}

\title{Theory of \'Atomata\tnoteref{support}}
\tnotetext[support]{This work was supported 
by the Natural Sciences and Engineering Research Council of Canada under grant No.~OGP0000871, 
the ERDF funded Estonian Center of Excellence in Computer Science, EXCS, 
the Estonian Science Foundation grant 7520, and the Estonian Ministry of Education and 
Research target-financed research theme no. 0140007s12.}

\author[JB]{Janusz~Brzozowski}
\ead{brzozo@uwaterloo.ca}
\author[HT]{Hellis~Tamm}
\ead{hellis@cs.ioc.ee}
\address[JB]{David R. Cheriton School of Computer Science, University of Waterloo,
Waterloo, ON, Canada N2L 3G1}
\address[HT]{Institute of Cybernetics, Tallinn University of Technology,
Akadeemia tee 21, 12618 Tallinn, Estonia}

\begin{abstract}
We show that every regular language defines a unique nondeterministic finite 
automaton (NFA), which we call ``\'atomaton'', whose states are the ``atoms'' of 
the language, that is, non-empty intersections of complemented or uncomplemented 
left quotients of the language.
We describe methods of constructing the \'atomaton, and prove that it is 
isomorphic to the reverse automaton of the minimal deterministic finite automaton 
(DFA) of the reverse language.
We study ``atomic'' NFAs in which the right language of every state 
is a union of atoms. We generalize Brzozowski's double-reversal method for 
minimizing a deterministic finite automaton (DFA), showing that the result of 
applying the subset construction to an NFA is a minimal DFA if and only if 
the reverse of the NFA is atomic.
We prove that Sengoku's claim that his method always finds a minimal NFA is false.
\end{abstract}

\begin{keyword}
regular languages \sep left quotients
\end{keyword}

\end{frontmatter}

\section{Introduction}

Nondeterministic finite automata (NFAs) were introduced by 
Rabin and Scott~\cite{RaSc59} in 1959 and still
play a major role in the theory of automata. 
For many purposes it is necessary to convert an NFA to a deterministic finite automaton (DFA).
In particular, for each NFA there exists a minimal DFA, unique up to isomorphism.
This DFA is uniquely defined by every regular language, and uses the left quotients 
of the language as states. As well, it is  possible to associate an NFA with each DFA, 
and this is the subject of the present paper. Our NFA is also uniquely defined by 
every regular language, and uses non-empty intersections of complemented and 
uncomplemented quotients---the ``atoms'' of the language---as states. 

It appears that the NFA most often associated with a regular language is 
the universal automaton, sometimes  appearing under different names. 
A substantial survey by Lombardy and Sakarovitch~\cite{LoSa07} on the subject 
of the universal automaton contains its history and a detailed discussion 
of its properties. We refer the reader to that paper, and mention only that 
research related to the universal automaton goes back to the 
1970's: \eg, in~\cite{Car70} as reported in~\cite{ADN92}, \cite{Con71,KaWe70}.

We call our NFA the ``\'atomaton'' because it is based on the atoms of a regular language; 
we add the accent\footnote{The word should be pronounced with 
the accent on the first a.} to minimize the possible confusion between ``automaton'' and ``atomaton''.
We prove that the \'atomaton of a regular language $L$ is isomorphic to the reverse 
automaton of the minimal DFA of the reverse language $L^R$.

We  introduce ``atomic'' automata, in which the right language of any state is 
a union of some atoms. 
This generalizes residual automata~\cite{DLT02} 
in which the right language of any state is a left quotient (which we prove 
to be a union of atoms), and includes also  \'atomata (where the right language 
of any state is an atom), DFAs, and universal automata.

We characterize the class of NFAs for which the subset construction 
yields a minimal DFA. More specifically, we show that the subset construction 
applied to an NFA produces a minimal DFA if and only if the reverse automaton 
of that NFA is atomic. 
This generalizes Brzozowski's method for DFA minimization by double reversal~\cite{Brz63}.

We study reduced atomic NFAs associated with a given regular language. We formalize Sengoku's approach~\cite{Sen92} to finding minimal NFAs, and prove that it does not always work, since there exist languages for which no atomic NFA is minimal.

Section~\ref{sec:LAE} recalls properties of regular languages, finite automata, and systems of language equations.
In Section~\ref{sec:partial}, we study the right languages of the states of any NFA; we call these languages ``partial quotients". We define partial atoms and partial \'atomata and study their properties.
Quotients and atoms of a regular language and the \'atomaton are introduced and 
studied in Section~\ref{sec:atoms}.
In Section~\ref{sec:atomic}, we examine NFAs in which the right language of every state is a union of atoms, and we extend
Brzozowski's method~\cite{Brz63} of DFA minimization.
In Section~\ref{sec:reduced}, we study reduced atomic NFAs accepting a given regular language, and prove that Sengoku's claim~\cite{Sen92} that his method always finds a minimal NFA is false.
Section~\ref{sec:conc} closes the paper.

A much shorter version of some of the results presented here has previously 
appeared in~\cite{BrTa11}. Note that here we use a definition of an atom 
which is slightly different from that of~\cite{BrTa11} for reasons explained 
at the end of Section~\ref{sec:atoms}.

\section{Languages, Automata and Equations}
\label{sec:LAE}
If $\Sig$ is a non-empty finite alphabet, then $\Sig^*$ is the free monoid generated by $\Sig$.
 A \emph{word} is any element of $\Sig^*$, and the empty word is $\eps$. The length of a word $w$ is $|w|$.
A \emph{language} over $\Sig$ is any subset of $\Sig^*$. 

The following   operations are defined on languages over $\Sig$:  
\emph{complement} \linebreak ($\ol{L}=\Sigma^*\setminus L$),  
\emph{union}  ($K\cup L$),  
\emph{intersection} ($K\cap L$),  
\emph{product}, usually called  
\emph{concatenation} or \emph{catenation},   
($KL=\{w\in \Sigma^*\mid w=uv, u\in K, v\in L\}$), \emph{positive closure} ($L^+=\bigcup_{i\ge 1}L^i$), and  \emph{star} ($L^*=\bigcup_{i\ge 0}L^i$).
The \emph{reverse $w^R$ of a word} $w\in\Sigma^*$ is defined as follows: 
$\eps^R=\eps$, and $(wa)^R=aw^R$, where $a\in\Sig$.
The \emph{reverse of a language} $L$ is denoted 
by $L^R$ and defined as $L^R=\{w^R\mid w\in L\}$.

A~\emph{nondeterministic finite automaton} is a quintuple 
$\cN=(Q, \Sig, \delta, I,F)$, where 
$Q$ is a finite, non-empty set of \emph{states}, 
$\Sig$ is a finite non-empty \emph{alphabet}, 
$\delta:Q\times \Sig\to 2^Q$ is the  \emph{transition function},
$I\subseteq  Q$ is the set of  \emph{initial states},
and $F\subseteq Q$ is the set of \emph{final states}.
As usual, we extend the transition function to functions 
$\delta':Q\times \Sig^*\to 2^Q$, and $\delta'':2^Q\times \Sig^*\to 2^Q$.
We do not distinguish these functions notationally, but use $\delta$ for all three.
The \emph{language accepted} by an NFA $\cN$ is 
$L(\cN)=\{w\in\Sig^*\mid \delta(I,w)\cap F\neq \emp\}$.
Two NFAs are \emph{equivalent} if they accept the same language. 
The  \emph{left language} of a state $q$ of $\cN$ is 
$L_{I,q}(\cN)=\{w\in\Sig^* \mid q\in \delta(I,w)\}$. 
The \emph{right language} of $q$ is 
$L_{q,F}(\cN)=\{w\in\Sig^* \mid \delta(q,w)\cap F\neq\emp\}$.
The \emph{right language} of a set $S$ of states of $\cN$ is
$L_{S,F}(\cN)=\bigcup_{q\in S} L_{q,F}(\cN)$; hence
$L(\cN)=L_{I,F}(\cN)$.
A state is \emph{unreachable} if its left language is empty.
A state is \emph{empty} if its right language is empty.
An NFA is \emph{trim} if it has no empty or unreachable states.
Two states  of an NFA are \emph{equivalent} if their right languages are identical. 
An NFA is \emph{reduced} if it has no equivalent states. 
An NFA is \emph{minimal} if it has the minimal number of states among all
the equivalent NFAs.

A \emph{deterministic finite automaton} is a quintuple 
$\cD=(Q, \Sig, \delta, q_0,F)$, where
$Q$, $\Sig$, and $F$ are as in an NFA, 
$\delta:Q\times \Sig\to Q$ is the transition function, 
and $q_0$ is the initial state. 
A DFA is an NFA in which the set of initial states is $\{q_0\}$ and 
the range of the transition function is restricted to singletons $\{q\}$, $q\in Q$.

A DFA is \emph{minimal} if it has no unreachable states and 
no two of its states are equivalent.

We use the following operations on automata: \\
\hglue 10 pt 1.
\emph{Determinization} ($\deter$) 
applied to an NFA $\cN$ yields a DFA $\cN^{\deter}$ obtained by the well-known subset construction, where only subsets (including the empty subset) reachable from the initial subset of $\cN^\deter$ are used. \\
\hglue 10 pt 2.
 \emph{Reversal} ($\rev$) applied to an NFA $\cN$ yields an NFA $\cN^{\rev}$, where initial and final states of $\cN$ 
are interchanged in $\cN^\rev$ and all the transitions  
are reversed. \\
\hglue 10 pt 3.
\emph{Trimming} ($\trim$) applied to an NFA $\cN$ accepting 
a non-empty language deletes from $\cN$ all unreachable and empty states,
along with the incident transitions, yielding an NFA $\cN^{\trim}$. \\
\hglue 10 pt 4.
\emph{Minimization}  ($\mini$) applied to a DFA $\cD$ yields 
the minimal DFA $\cD^{\mini}$ equivalent to $\cD$.

A trim DFA $\cD$ is \emph{bideterministic} if also $\cD^{\rev}$ is a DFA.
A language is \emph{bideterministic} if it is accepted by a bideterministic 
DFA.

\begin{example}
\label{ex:aut_ops}
Figure~\ref{fig:aut_ops}~(a) shows an NFA $\cN$, where the initial states are indicated by incoming arrows and final states by double circles. Its determinized DFA $\cN^{\deter}$ is in 
Fig.~\ref{fig:aut_ops}~(b), where braces around sets are omitted.  
The minimal equivalent DFA $\cD=\cN^{\deter\mini}$   of $\cN^{\deter}$  
is in Fig.~\ref{fig:aut_ops}~(c),
where the equivalent states $\{1\}$, $\{0,2\}$,  and $\{1,2\}$ are represented by 
$\{0,2\}$. The reversed and trimmed version $\cD^{\rev\trim}=\cN^{\deter\mini\rev\trim}$ of $\cD$ 
is in Fig.~\ref{fig:aut_ops}~(d). \qedb
\end{example}
\begin{figure}[hbt]
\begin{center}
\setlength{\unitlength}{0.00043745in}
\begingroup\makeatletter\ifx\SetFigFont\undefined%
\gdef\SetFigFont#1#2#3#4#5{%
  \reset@font\fontsize{#1}{#2pt}%
  \fontfamily{#3}\fontseries{#4}\fontshape{#5}%
  \selectfont}%
\fi\endgroup%
{\renewcommand{\dashlinestretch}{30}
\begin{picture}(10284,4052)(0,-10)
\put(8183,1622){\makebox(0,0)[lb]{\smash{{\SetFigFont{10}{12.0}{\familydefault}{\mddefault}{\updefault}$0$}}}}
\put(5186.000,3444.000){\arc{364.664}{2.4959}{6.9871}}
\blacken\path(5339.259,3448.869)(5325.000,3326.000)(5395.402,3427.703)(5339.259,3448.869)
\put(8262.000,3455.000){\arc{364.664}{2.4959}{6.9871}}
\blacken\path(8415.259,3459.869)(8401.000,3337.000)(8471.402,3438.703)(8415.259,3459.869)
\put(1196.691,3090.686){\arc{364.730}{2.4963}{6.9846}}
\blacken\path(1350.223,3095.873)(1336.000,2973.000)(1406.372,3074.724)(1350.223,3095.873)
\put(3724.691,3082.686){\arc{364.730}{2.4963}{6.9846}}
\blacken\path(3878.223,3087.873)(3864.000,2965.000)(3934.372,3066.724)(3878.223,3087.873)
\put(10009.686,1740.000){\arc{364.258}{4.0681}{8.5565}}
\blacken\path(10014.856,1586.630)(9892.000,1601.000)(9993.640,1530.506)(10014.856,1586.630)
\put(3316.686,576.000){\arc{364.258}{4.0681}{8.5565}}
\blacken\path(3321.856,422.630)(3199.000,437.000)(3300.640,366.506)(3321.856,422.630)
\put(5182,1706){\ellipse{720}{720}}
\put(5186,440){\ellipse{720}{720}}
\put(6446,1707){\ellipse{720}{720}}
\put(6453,454){\ellipse{720}{720}}
\put(5184,437){\ellipse{630}{630}}
\put(5195,2999){\ellipse{720}{720}}
\put(8265,1729){\ellipse{720}{720}}
\put(9529,1730){\ellipse{720}{720}}
\put(8269,3005){\ellipse{720}{720}}
\put(1194,2635){\ellipse{720}{720}}
\put(1194,2635){\ellipse{630}{630}}
\put(3718,2637){\ellipse{720}{720}}
\put(2458,2638){\ellipse{720}{720}}
\put(3718,2637){\ellipse{630}{630}}
\put(1583,590){\ellipse{720}{720}}
\put(2842,588){\ellipse{630}{630}}
\put(2843,589){\ellipse{720}{720}}
\put(6441,1709){\ellipse{630}{630}}
\put(6452,452){\ellipse{630}{630}}
\put(9536,1726){\ellipse{630}{630}}
\path(5550,446)(6082,446)
\blacken\path(5962.000,416.000)(6082.000,446.000)(5962.000,476.000)(5962.000,416.000)
\path(6450,1339)(6450,807)
\blacken\path(6420.000,927.000)(6450.000,807.000)(6480.000,927.000)(6420.000,927.000)
\path(5003,1384)(5003,754)
\blacken\path(4973.000,874.000)(5003.000,754.000)(5033.000,874.000)(4973.000,874.000)
\path(5363,751)(5363,1381)
\blacken\path(5393.000,1261.000)(5363.000,1381.000)(5333.000,1261.000)(5393.000,1261.000)
\path(6188,723)(5445,1451)
\blacken\path(5551.709,1388.445)(5445.000,1451.000)(5509.718,1345.588)(5551.709,1388.445)
\path(6085,2459)(6333,2038)
\blacken\path(6246.245,2126.168)(6333.000,2038.000)(6297.942,2156.621)(6246.245,2126.168)
\path(5190,2081)(5190,2621)
\blacken\path(5220.000,2501.000)(5190.000,2621.000)(5160.000,2501.000)(5220.000,2501.000)
\path(9168,2482)(9416,2061)
\blacken\path(9329.245,2149.168)(9416.000,2061.000)(9380.942,2179.621)(9329.245,2149.168)
\path(8258,2099)(8258,2639)
\blacken\path(8288.000,2519.000)(8258.000,2639.000)(8228.000,2519.000)(8288.000,2519.000)
\path(1509,2815)(2139,2815)
\blacken\path(2019.000,2785.000)(2139.000,2815.000)(2019.000,2845.000)(2019.000,2785.000)
\path(2139,2455)(1509,2455)
\blacken\path(1629.000,2485.000)(1509.000,2455.000)(1629.000,2425.000)(1629.000,2485.000)
\path(2004,3348)(2252,2927)
\blacken\path(2165.245,3015.168)(2252.000,2927.000)(2216.942,3045.621)(2165.245,3015.168)
\path(3272,3329)(3520,2908)
\blacken\path(3433.245,2996.168)(3520.000,2908.000)(3484.942,3026.621)(3433.245,2996.168)
\path(3353,2633)(2813,2633)
\blacken\path(2933.000,2663.000)(2813.000,2633.000)(2933.000,2603.000)(2933.000,2663.000)
\path(6075,1698)(5565,1698)
\blacken\path(5685.000,1728.000)(5565.000,1698.000)(5685.000,1668.000)(5685.000,1728.000)
\path(9234,1923)(8604,1923)
\blacken\path(8724.000,1953.000)(8604.000,1923.000)(8724.000,1893.000)(8724.000,1953.000)
\path(8589,1563)(9219,1563)
\blacken\path(9099.000,1533.000)(9219.000,1563.000)(9099.000,1593.000)(9099.000,1533.000)
\path(2397,1281)(2645,860)
\blacken\path(2558.245,948.168)(2645.000,860.000)(2609.942,978.621)(2558.245,948.168)
\path(2540,400)(1910,400)
\blacken\path(2030.000,430.000)(1910.000,400.000)(2030.000,370.000)(2030.000,430.000)
\path(1887,768)(2517,768)
\blacken\path(2397.000,738.000)(2517.000,768.000)(2397.000,798.000)(2397.000,738.000)
\path(1276,2283)(1277,2282)(1279,2281)
	(1283,2278)(1289,2273)(1298,2267)
	(1310,2259)(1325,2249)(1342,2236)
	(1363,2222)(1387,2206)(1414,2188)
	(1443,2169)(1475,2149)(1509,2128)
	(1544,2107)(1581,2085)(1620,2063)
	(1660,2042)(1701,2020)(1744,1999)
	(1788,1979)(1833,1959)(1880,1941)
	(1929,1923)(1980,1906)(2033,1891)
	(2088,1877)(2146,1864)(2206,1853)
	(2269,1844)(2334,1837)(2401,1833)
	(2469,1832)(2537,1834)(2603,1838)
	(2667,1845)(2728,1855)(2787,1866)
	(2842,1879)(2895,1893)(2945,1909)
	(2993,1926)(3039,1944)(3083,1963)
	(3125,1983)(3166,2004)(3205,2025)
	(3243,2046)(3279,2068)(3314,2090)
	(3347,2112)(3379,2134)(3409,2155)
	(3438,2175)(3464,2194)(3488,2212)
	(3509,2228)(3527,2243)(3543,2255)
	(3556,2266)(3567,2274)(3574,2280)(3586,2290)
\blacken\path(3513.019,2190.131)(3586.000,2290.000)(3474.608,2236.225)(3513.019,2190.131)
\put(7149,349){\makebox(0,0)[lb]{\smash{{\SetFigFont{10}{12.0}{\familydefault}{\mddefault}{\updefault}$b$}}}}
\put(5438,859){\makebox(0,0)[lb]{\smash{{\SetFigFont{10}{12.0}{\familydefault}{\mddefault}{\updefault}$a$}}}}
\put(5881,1099){\makebox(0,0)[lb]{\smash{{\SetFigFont{10}{12.0}{\familydefault}{\mddefault}{\updefault}$a$}}}}
\put(6547,940){\makebox(0,0)[lb]{\smash{{\SetFigFont{10}{12.0}{\familydefault}{\mddefault}{\updefault}$b$}}}}
\put(5107,333){\makebox(0,0)[lb]{\smash{{\SetFigFont{10}{12.0}{\familydefault}{\mddefault}{\updefault}$1$}}}}
\put(5092,1592){\makebox(0,0)[lb]{\smash{{\SetFigFont{10}{12.0}{\familydefault}{\mddefault}{\updefault}$0$}}}}
\put(4695,880){\makebox(0,0)[lb]{\smash{{\SetFigFont{10}{12.0}{\familydefault}{\mddefault}{\updefault}$b$}}}}
\put(5723,123){\makebox(0,0)[lb]{\smash{{\SetFigFont{10}{12.0}{\familydefault}{\mddefault}{\updefault}$b$}}}}
\put(5108,2906){\makebox(0,0)[lb]{\smash{{\SetFigFont{10}{12.0}{\familydefault}{\mddefault}{\updefault}$\emp$}}}}
\put(5295,2223){\makebox(0,0)[lb]{\smash{{\SetFigFont{10}{12.0}{\familydefault}{\mddefault}{\updefault}$a$}}}}
\put(8363,2250){\makebox(0,0)[lb]{\smash{{\SetFigFont{10}{12.0}{\familydefault}{\mddefault}{\updefault}$a$}}}}
\put(1104,3363){\makebox(0,0)[lb]{\smash{{\SetFigFont{10}{12.0}{\familydefault}{\mddefault}{\updefault}$b$}}}}
\put(3623,3348){\makebox(0,0)[lb]{\smash{{\SetFigFont{10}{12.0}{\familydefault}{\mddefault}{\updefault}$b$}}}}
\put(1682,2898){\makebox(0,0)[lb]{\smash{{\SetFigFont{10}{12.0}{\familydefault}{\mddefault}{\updefault}$a$}}}}
\put(1771,2170){\makebox(0,0)[lb]{\smash{{\SetFigFont{10}{12.0}{\familydefault}{\mddefault}{\updefault}$b$}}}}
\put(2326,1923){\makebox(0,0)[lb]{\smash{{\SetFigFont{10}{12.0}{\familydefault}{\mddefault}{\updefault}$b$}}}}
\put(3001,2371){\makebox(0,0)[lb]{\smash{{\SetFigFont{10}{12.0}{\familydefault}{\mddefault}{\updefault}$a$}}}}
\put(5730,1810){\makebox(0,0)[lb]{\smash{{\SetFigFont{10}{12.0}{\familydefault}{\mddefault}{\updefault}$a$}}}}
\put(10269,1640){\makebox(0,0)[lb]{\smash{{\SetFigFont{10}{12.0}{\familydefault}{\mddefault}{\updefault}$b$}}}}
\put(7981,3780){\makebox(0,0)[lb]{\smash{{\SetFigFont{10}{12.0}{\familydefault}{\mddefault}{\updefault}$a,b$}}}}
\put(4906,3782){\makebox(0,0)[lb]{\smash{{\SetFigFont{10}{12.0}{\familydefault}{\mddefault}{\updefault}$a,b$}}}}
\put(8852,2036){\makebox(0,0)[lb]{\smash{{\SetFigFont{10}{12.0}{\familydefault}{\mddefault}{\updefault}$a$}}}}
\put(8844,1258){\makebox(0,0)[lb]{\smash{{\SetFigFont{10}{12.0}{\familydefault}{\mddefault}{\updefault}$b$}}}}
\put(2088,870){\makebox(0,0)[lb]{\smash{{\SetFigFont{10}{12.0}{\familydefault}{\mddefault}{\updefault}$a$}}}}
\put(2081,100){\makebox(0,0)[lb]{\smash{{\SetFigFont{10}{12.0}{\familydefault}{\mddefault}{\updefault}$b$}}}}
\put(3580,468){\makebox(0,0)[lb]{\smash{{\SetFigFont{10}{12.0}{\familydefault}{\mddefault}{\updefault}$b$}}}}
\put(8184,2910){\makebox(0,0)[lb]{\smash{{\SetFigFont{10}{12.0}{\familydefault}{\mddefault}{\updefault}$\emp$}}}}
\put(9240,2913){\makebox(0,0)[lb]{\smash{{\SetFigFont{10}{12.0}{\familydefault}{\mddefault}{\updefault}{\bf (c)}}}}}
\put(6135,2913){\makebox(0,0)[lb]{\smash{{\SetFigFont{10}{12.0}{\familydefault}{\mddefault}{\updefault}{\bf (b)}}}}}
\put(15,2553){\makebox(0,0)[lb]{\smash{{\SetFigFont{10}{12.0}{\familydefault}{\mddefault}{\updefault}{\bf (a)}}}}}
\put(15,483){\makebox(0,0)[lb]{\smash{{\SetFigFont{10}{12.0}{\familydefault}{\mddefault}{\updefault}{\bf (d)}}}}}
\put(1118,2545){\makebox(0,0)[lb]{\smash{{\SetFigFont{10}{12.0}{\familydefault}{\mddefault}{\updefault}$1$}}}}
\put(2379,2545){\makebox(0,0)[lb]{\smash{{\SetFigFont{10}{12.0}{\familydefault}{\mddefault}{\updefault}$0$}}}}
\put(3631,2546){\makebox(0,0)[lb]{\smash{{\SetFigFont{10}{12.0}{\familydefault}{\mddefault}{\updefault}$2$}}}}
\put(1511,482){\makebox(0,0)[lb]{\smash{{\SetFigFont{10}{12.0}{\familydefault}{\mddefault}{\updefault}$0$}}}}
\put(2602,498){\makebox(0,0)[lb]{\smash{{\SetFigFont{10}{12.0}{\familydefault}{\mddefault}{\updefault}$0,2$}}}}
\put(6202,1595){\makebox(0,0)[lb]{\smash{{\SetFigFont{10}{12.0}{\familydefault}{\mddefault}{\updefault}$0,2$}}}}
\put(6210,348){\makebox(0,0)[lb]{\smash{{\SetFigFont{10}{12.0}{\familydefault}{\mddefault}{\updefault}$1,2$}}}}
\put(9286,1632){\makebox(0,0)[lb]{\smash{{\SetFigFont{10}{12.0}{\familydefault}{\mddefault}{\updefault}$0,2$}}}}
\put(6897.686,464.000){\arc{364.258}{4.0681}{8.5565}}
\blacken\path(6902.856,310.630)(6780.000,325.000)(6881.640,254.506)(6902.856,310.630)
\end{picture}
}
\end{center}
\caption{(a) An NFA $\cN$; (b) $\cN^{\deter}$; (c) $\cN^{\deter\mini}$; 
(d) $\cN^{\deter\mini\rev\trim}$.} 
\label{fig:aut_ops}
\end{figure}

The \emph{left quotient}, or simply \emph{quotient,} of a language $L$ 
by a word $w$ is  the language $w^{-1}L=\{x\in \Sig^*\mid wx\in L \}$. 
Left quotients are also known as \emph{right residuals}.
Dually, the \emph{right quotient} of a language $L$ by a word $w$ is  
the language $Lw^{-1}=\{x\in \Sig^*\mid xw\in L \}$. 
Evidently, if $\cN$ is an NFA and $x$ is in $L_{I,q}(\cN)$,  
then $L_{q,F}(\cN)\subseteq x^{-1}(L(\cN))$.

The \emph{quotient DFA} of a regular language $L$ is 
$\cD=(Q, \Sig, \delta, q_0,F)$, where $Q=\{w^{-1}L\mid w\in\Sig^*\}$, 
$\delta(w^{-1}L,a)=a^{-1}(w^{-1}L)$, 
$q_0=\eps^{-1}L=L$,  and
$F=\{w^{-1}L \mid \eps\in w^{-1}L\}$.
The quotient DFA of $L$ is the minimal DFA for $L$.

The following is from~\cite{LoSa07}:
If $L\subseteq \Sig^*$, a \emph{subfactorization} 
of $L$ is  a pair $(X, Y)$ of languages 
over $\Sig$ such that $X Y \subseteq L $.
A~\emph{factorization} of $L$ is 
a subfactorisation $(X, Y)$ such that, if $X \subseteq X'$,  $Y \subseteq Y'$,
and $X'Y' \subseteq L$ for any pair $(X',Y')$, then $X = X'$ and $Y = Y'$.
The \emph{universal automaton} of $L$ is 
$\cU_L=(Q,\Sigma,\delta,I,F)$ where $Q$ is the set of all factorizations
of $L$, $I=\{(X,Y)\in Q \;|\; \eps\in X\}$,
$F=\{(X,Y)\in Q \;|\; \eps\in Y\}$, and
$(X',Y')\in\delta((X,Y),a)$ if and only if $XaY'\subseteq L$.

For any language $L$ let $L^\eps=\emp$ if $\eps\not\in L$ and $L^\eps=\{\eps\}$
otherwise.
Also, let $n\ge 1$ and let $[n]=\{0,\ldots,n-1\}$.
A~\emph{nondeterministic system of equations (NSE)} with $n$ variables 
$L_0,\ldots, L_{n-1}$ is a set of language equations \\
\begin{equation}
L_i = \bigcup_{a\in \Sig} a(\bigcup_{j\in J_{i,a}} L_j) \cup L_i^\eps \quad i=0,\ldots,n-1,
\end{equation}
where $J_{i,a}\subseteq [n]$,  together with an 
\emph{initial set of variables}  $\{L_i\mid i\in I\}$, where $I\subseteq [n]$ is an index set.
The equations are assumed to have been simplified by the rules
$a\emp=\emp \text{ and } K\cup \emp= \emp\cup K=K, \text{ for any language } K.$
Let $L_{i,a}=\bigcup_{j\in J_{i,a}} L_j$; then $L_{i,a}=a^{-1}L_i$ is the left quotient of $L_i$ by $a$.
The language defined by an NSE is $L=\bigcup_{i\in I} L_i$.

Each NSE defines a unique NFA $\cN$ and \emph{vice versa}. States of $\cN$ correspond 
to the variables $L_i$, there is a transition  $L_i \stackrel{a}{\rightarrow} L_j$ 
in $\cN$ if and only if $j\in J_{i,a}$, the set of initial states of $\cN$ is 
$\{L_i \mid i\in I\}$, and the set of final states is $\{L_i \mid L_i^\eps=\{\eps\}\}$.

If each $L_i$ is a left quotient (that is, a right residual) of the language 
$L=\bigcup_{i\in I} L_i$, then the NSE and the corresponding NFA are called 
\emph{residual}~\cite{DLT02}.

A~\emph{deterministic system of equations (DSE)} 
is an NSE
\begin{equation}
L_i = \bigcup_{a\in \Sig} aL_{i_a} \cup L_i^\eps \quad i=0,\ldots,n-1,
\end{equation}
where ${i_a}\in [n]$, $I=\{0\}$, and 
the empty language $\emp$ is retained if it appears. 

Each DSE defines a unique DFA $\cD$ and \emph{vice versa}. Each state of $\cD$ 
corresponds to a variable $L_i$, there is a transition  
$L_i \stackrel{a}{\rightarrow} L_j$ in $\cD$ if and only if 
$i_a=j$,
the initial state of $\cD$ corresponds to $L_0$, and the set of final states is 
$\{L_i \mid L_i^\eps=\{\eps\}\}$.
In the special case when $\cD$ is minimal, its DSE constitutes its 
\emph{quotient equations}, where every $L_i$ is a quotient of the initial 
language $L_0$.

To simplify the notation, we write $\eps$ instead of $\{\eps\}$ in  equations.
\begin{example}
\label{ex:NSE}
For the NFA of Fig.~\ref{fig:aut_ops}~(a), we have the NSE 
\begin{displaymath}
 \begin{array}{rl}
   L_0&=  bL_1, \\
   L_1&= aL_0 \cup b(L_1 \cup L_2)   \cup \eps, \\
   L_2&= aL_0\cup bL_2 \cup \eps, 
 \end{array}
\end{displaymath}
with the initial set 
$\{L_0,L_2\}$.
The language $L=L_0\cup L_2$ accepted by the DFA of Fig.~\ref{fig:aut_ops}~(b) 
is obtained  from this NSE as shown by the equations below on the left. 
Renaming the unions of variables by new variables corresponding to subsets in 
the subset construction, we get the equations on the right;
for example, $L_0\cup L_2$ is renamed as $L_{\{0,2\}}$.
This is
the DSE for the DFA of Fig.~\ref{fig:aut_ops}~(b). 
\begin{alignat*}{2}
 L_0\cup L_2&= aL_0 \cup b(L_1\cup L_2) \cup\eps,
	\qquad 	&L_{\{0,2\}}&= aL_{\{0\}} \cup b L_{\{1,2\}} \cup \eps, \\
 L_0&= a\emp\cup bL_1, & 
 	L_{\{0\}}&= aL_{\emp}\cup bL_{\{1\}},\\
 L_1\cup L_2&= aL_0 \cup b(L_1\cup L_2)  \cup \eps, &
 	 L_{\{1,2\}}&= aL_{\{0\}} \cup bL_{\{1,2\}}  \cup \eps, \\
L_1&= aL_0\cup b(L_1\cup L_2) \cup \eps, &
	L_{\{1\}}&= aL_{\{0\}}\cup bL_{\{1,2\}} \cup \eps,\\
\emp&= a\emp\cup b\emp.  &    L_{\emp}&= aL_{\emp}\cup bL_{\emp}.  
		\end{alignat*}
 Noting that $L_{\{0,2\}}$, $L_{\{1,2\}}$, and $L_{\{1\}}$ are equivalent, 
we get the quotient equations for the DFA of Fig.~\ref{fig:aut_ops}~(c), 
where $L_{\{0\}}=a^{-1}L_{\{0,2\}}$, $L_{\{0,2\}}=b^{-1}L_{\{0,2\}}$, \etc
\begin{displaymath}
 \begin{array}{rl}
   L_{\{0,2\}}&= aL_{\{0\}} \cup b L_{\{0,2\}} \cup \eps, \\
     L_{\{0\}}&= aL_{\emp}\cup bL_{\{0,2\}}, \\
   L_{\emp}&= aL_{\emp}\cup bL_{\emp}. 
 \end{array}  
\end{displaymath}  
\qedb\end{example}

\section{Partial Quotients, Partial Atoms and Partial \'Atomata}
\label{sec:partial}
\subsection{Introduction and Motivation}
In 1992, Sengoku~\cite{Sen92} defined an NFA to be \emph{disjoint} if the right languages of any two distinct states are disjoint. He noted that a disjoint NFA $\cN$ has exactly one final state, and proved that an NFA is disjoint if and only if $\cN^R$ is deterministic. It follows that if we reverse, determinize\footnote{The reader should note that Sengoku's DFAs are incomplete, and he does not include the empty
state in the  determinized version of an NFA.}, and reverse $\cN$, the resulting NFA $\cN^{RDR}$ is a disjoint NFA equivalent to $\cN$. 

We shall show that another NFA obtained from $\cN$ by a completely different process turns out to be isomorphic to $\cN^{RDR}$. In our approach we start with the right languages of states of $\cN$, which we call \emph{partial quotients of $\cN$}. This terminology is logical, since the right language of a state of $\cN$ that is reached by word $w\in\Sig^*$ is always a subset of the quotient of $L(\cN)$ by $w$. Next we construct all nonempty intersections of complemented and uncomplemented partial quotients of $\cN$, and refer to these languages as 
\emph{partial atoms of $\cN$}. These partial atoms become states of an NFA which we call 
\emph{partial \'atomaton of $\cN$}. We then prove that the partial \'atomaton of $\cN$ is isomorphic to $\cN^{RDR}$.

We begin with a simple example to illustrate the formal ideas that follow.

\begin{figure}[t]
\begin{center}
\setlength{\unitlength}{0.00043745in}
\begingroup\makeatletter\ifx\SetFigFont\undefined%
\gdef\SetFigFont#1#2#3#4#5{%
  \reset@font\fontsize{#1}{#2pt}%
  \fontfamily{#3}\fontseries{#4}\fontshape{#5}%
  \selectfont}%
\fi\endgroup%
{\renewcommand{\dashlinestretch}{30}
\begin{picture}(7860,3207)(0,-10)
\put(7332,1540){\makebox(0,0)[lb]{\smash{{\SetFigFont{10}{12.0}{\familydefault}{\mddefault}{\updefault}$02$}}}}
\put(4973.000,2107.000){\arc{364.664}{2.4959}{6.9871}}
\blacken\thicklines
\path(5118.214,2128.974)(5112.000,1989.000)(5189.516,2105.715)(5118.214,2128.974)
\thinlines
\put(2085.000,3000.000){\arc{364.664}{2.4959}{6.9871}}
\blacken\thicklines
\path(2230.214,3021.974)(2224.000,2882.000)(2301.516,2998.715)(2230.214,3021.974)
\thinlines
\put(2079.309,243.314){\arc{364.730}{5.6379}{10.1262}}
\blacken\thicklines
\path(1934.779,220.986)(1940.000,361.000)(1863.313,243.738)(1934.779,220.986)
\thinlines
\put(6237,737){\ellipse{720}{720}}
\put(6237,2530){\ellipse{720}{720}}
\put(6235,739){\ellipse{630}{630}}
\put(7492,1644){\ellipse{720}{720}}
\put(815,1630){\ellipse{720}{720}}
\put(2082,730){\ellipse{720}{720}}
\put(2082,2523){\ellipse{720}{720}}
\put(2080,732){\ellipse{630}{630}}
\put(4970,1637){\ellipse{720}{720}}
\path(6057,2222)(6057,1097)
\blacken\thicklines
\path(6019.500,1232.000)(6057.000,1097.000)(6094.500,1232.000)(6019.500,1232.000)
\thinlines
\path(6417,1052)(6417,2177)
\blacken\thicklines
\path(6454.500,2042.000)(6417.000,2177.000)(6379.500,2042.000)(6454.500,2042.000)
\thinlines
\path(6556,2314)(7186,1864)
\blacken\thicklines
\path(7054.350,1911.952)(7186.000,1864.000)(7097.942,1972.982)(7054.350,1911.952)
\blacken\path(5378.650,1319.048)(5247.000,1367.000)(5335.058,1258.018)(5378.650,1319.048)
\thinlines
\path(5247,1367)(5877,917)
\path(5015,857)(5015,1262)
\blacken\thicklines
\path(5052.500,1127.000)(5015.000,1262.000)(4977.500,1127.000)(5052.500,1127.000)
\thinlines
\path(7543,864)(7543,1269)
\blacken\thicklines
\path(7580.500,1134.000)(7543.000,1269.000)(7505.500,1134.000)(7580.500,1134.000)
\thinlines
\path(12,1630)(417,1630)
\blacken\thicklines
\path(282.000,1592.500)(417.000,1630.000)(282.000,1667.500)(282.000,1592.500)
\thinlines
\path(1902,2215)(1902,1090)
\blacken\thicklines
\path(1864.500,1225.000)(1902.000,1090.000)(1939.500,1225.000)(1864.500,1225.000)
\thinlines
\path(2262,1045)(2262,2170)
\blacken\thicklines
\path(2299.500,2035.000)(2262.000,2170.000)(2224.500,2035.000)(2299.500,2035.000)
\thinlines
\path(1092,1360)(1722,910)
\blacken\thicklines
\path(1590.350,957.952)(1722.000,910.000)(1633.942,1018.982)(1590.350,957.952)
\thinlines
\path(1100,1847)(1737,2289)
\blacken\thicklines
\path(1647.464,2181.229)(1737.000,2289.000)(1604.708,2242.849)(1647.464,2181.229)
\thinlines
\path(5283,1865)(5920,2307)
\blacken\thicklines
\path(5830.464,2199.229)(5920.000,2307.000)(5787.708,2260.849)(5830.464,2199.229)
\thinlines
\path(7230,1373)(6593,931)
\blacken\thicklines
\path(6682.536,1038.771)(6593.000,931.000)(6725.292,977.151)(6682.536,1038.771)
\put(6147,647){\makebox(0,0)[lb]{\smash{{\SetFigFont{10}{12.0}{\familydefault}{\mddefault}{\updefault}$2$}}}}
\put(6162,2417){\makebox(0,0)[lb]{\smash{{\SetFigFont{10}{12.0}{\familydefault}{\mddefault}{\updefault}$1$}}}}
\put(6530,1494){\makebox(0,0)[lb]{\smash{{\SetFigFont{10}{12.0}{\familydefault}{\mddefault}{\updefault}$b$}}}}
\put(5323,2192){\makebox(0,0)[lb]{\smash{{\SetFigFont{10}{12.0}{\familydefault}{\mddefault}{\updefault}$a$}}}}
\put(4760,2380){\makebox(0,0)[lb]{\smash{{\SetFigFont{10}{12.0}{\familydefault}{\mddefault}{\updefault}$a$}}}}
\put(7242,2364){\makebox(0,0)[lb]{\smash{{\SetFigFont{10}{12.0}{\familydefault}{\mddefault}{\updefault}$a$}}}}
\put(6889,827){\makebox(0,0)[lb]{\smash{{\SetFigFont{10}{12.0}{\familydefault}{\mddefault}{\updefault}$a$}}}}
\put(4827,205){\makebox(0,0)[lb]{\smash{{\SetFigFont{10}{12.0}{\familydefault}{\mddefault}{\updefault}{\bf (b)}}}}}
\put(5727,1509){\makebox(0,0)[lb]{\smash{{\SetFigFont{10}{12.0}{\familydefault}{\mddefault}{\updefault}$b$}}}}
\put(5284,932){\makebox(0,0)[lb]{\smash{{\SetFigFont{10}{12.0}{\familydefault}{\mddefault}{\updefault}$b$}}}}
\put(6793,2146){\makebox(0,0)[lb]{\smash{{\SetFigFont{10}{12.0}{\familydefault}{\mddefault}{\updefault}$b$}}}}
\put(732,1540){\makebox(0,0)[lb]{\smash{{\SetFigFont{10}{12.0}{\familydefault}{\mddefault}{\updefault}$0$}}}}
\put(192,145){\makebox(0,0)[lb]{\smash{{\SetFigFont{10}{12.0}{\familydefault}{\mddefault}{\updefault}{\bf (a)}}}}}
\put(1992,640){\makebox(0,0)[lb]{\smash{{\SetFigFont{10}{12.0}{\familydefault}{\mddefault}{\updefault}$2$}}}}
\put(2352,100){\makebox(0,0)[lb]{\smash{{\SetFigFont{10}{12.0}{\familydefault}{\mddefault}{\updefault}$a$}}}}
\put(2007,2410){\makebox(0,0)[lb]{\smash{{\SetFigFont{10}{12.0}{\familydefault}{\mddefault}{\updefault}$1$}}}}
\put(2375,1487){\makebox(0,0)[lb]{\smash{{\SetFigFont{10}{12.0}{\familydefault}{\mddefault}{\updefault}$b$}}}}
\put(1572,1502){\makebox(0,0)[lb]{\smash{{\SetFigFont{10}{12.0}{\familydefault}{\mddefault}{\updefault}$b$}}}}
\put(1168,2185){\makebox(0,0)[lb]{\smash{{\SetFigFont{10}{12.0}{\familydefault}{\mddefault}{\updefault}$a$}}}}
\put(1144,925){\makebox(0,0)[lb]{\smash{{\SetFigFont{10}{12.0}{\familydefault}{\mddefault}{\updefault}$a$}}}}
\put(1587,2927){\makebox(0,0)[lb]{\smash{{\SetFigFont{10}{12.0}{\familydefault}{\mddefault}{\updefault}$a$}}}}
\put(4819,1533){\makebox(0,0)[lb]{\smash{{\SetFigFont{10}{12.0}{\familydefault}{\mddefault}{\updefault}$01$}}}}
\thinlines
\put(7486.000,2129.000){\arc{364.664}{2.4959}{6.9871}}
\blacken\thicklines
\path(7631.214,2150.974)(7625.000,2011.000)(7702.516,2127.715)(7631.214,2150.974)
\end{picture}
}
\end{center}
\caption{(a) An NFA $\cN$; (b) partial \'atomaton $\cX$ of $\cN$.} 
\label{fig:partatom}
\end{figure}

\begin{example}
\label{ex:partial}
Consider the NFA $\cN$ of Fig.~\ref{fig:partatom}~(a) recognizing a language $L$ 
over alphabet $\Sig=\{a,b\}$. The right language of state $1$ is $L_1$, 
which is the set of all words having an odd number of $b$s, 
and that of state $2$ is $L_2=\ol{L_1}$, which is the set of all words 
having an even number of $b$s. It follows that  $L_0=a(L_1\cup L_2)=a\Sig^*$.
Each language $L_i$ is a partial quotient of $\cN$, 
since $L_0=L=\eps^{-1}L$, and $L_1, L_2\subseteq a^{-1}L$.
The NSE for $\cN$ is the set of equations 
\begin{displaymath}
 \begin{array}{rl}
L_0&= a(L_1 \cup L_2), \\
L_1&= aL_1 \cup bL_2, \\
L_2&= aL_2 \cup bL_1 \cup \eps, 
 \end{array}
\end{displaymath}
with initial set $\{L_0\}$.

Next, we construct the partial atoms of $\cN$. 
Since $L_2=\ol{L_1}$, the intersections containing $L_1\cap L_2$ and 
$\ol{L_1}\cap \ol{L_2}$ are empty. 
We also note that $L_1\cap \ol{L_2}=L_1$ and $\ol{L_1}\cap L_2=L_2$. 
The partial atoms are the non-empty intersections 
$X_0=L_0\cap L_1\cap \ol{L_2}$, 
$X_1=L_0\cap \ol{L_1}\cap L_2$, 
$X_2=\ol{L_0}\cap L_1\cap \ol{L_2}$, and 
$X_3=\ol{L_0}\cap \ol{L_1}\cap L_2$, and 
they obey the following equations:
\begin{displaymath}
 \begin{array}{rl}
X_0&=L_0\cap L_1\cap \ol{L_2} = aL_1=a(L_0\cap L_1\cap \ol{L_2}) \cup a(\ol{L_0}\cap L_1\cap \ol{L_2}),\\  
X_1&=L_0\cap \ol{L_1}\cap L_2 = aL_2= a(L_0\cap \ol{L_1}\cap L_2) \cup a(\ol{L_0}\cap \ol{L_1}\cap L_2), \\
X_2&= \ol{L_0}\cap L_1\cap \ol{L_2} = bL_2= b(L_0\cap \ol{L_1}\cap L_2) \cup b(\ol{L_0}\cap \ol{L_1}\cap L_2),\\
X_3&= \ol{L_0} \cap \ol{L_1}\cap L_2 = bL_1 \cup \eps=b(L_0\cap L_1\cap \ol{L_2}) \cup b(\ol{L_0}\cap L_1\cap \ol{L_2})\cup \eps.
 \end{array}
\end{displaymath}

If we identify $X_0=L_0\cap L_1\cap \ol{L_2}$ with $01$, 
$X_1=L_0\cap \ol{L_1}\cap L_2$ with $02$, 
$X_2=\ol{L_0}\cap L_1\cap \ol{L_2}$ with $1$, and 
$X_3=\ol{L_0} \cap \ol{L_1}\cap L_2$ with $2$, 
and use $\{X_0, X_1\}$ as the initial set,
we obtain the partial \'atomaton
$\cX$ of Fig.~\ref{fig:partatom}~(b).

From Fig.~\ref{fig:partatom}, we find that $X_0$ ($X_1$) is the set of all words that begin with $a$ and have an odd (even) number of $b$s.
Also, $X_2$ is the set of all words that begin with $b$ and have an odd number of $b$s, and a word is in $X_3$ if it is empty or begins with $b$ and has an even number of $b$s.

Now we construct $\cN^{RDR}$. The steps are shown in Tables~\ref{tab:N}--\ref{tab:Nrdr}, 
where initial (final) states are denoted by right (left) arrows.
Note that Table~\ref{tab:Nrdr} corresponds precisely to Fig.~\ref{fig:partatom}~(b).
\qedb
\end{example}

\begin{table}[hbt]
\begin{minipage}[b]{0.45\linewidth}
\caption{NFA $\cN$.}
\label{tab:N}
{\footnotesize
\begin{center}
$
\begin{array}{|c| c||c| c| c|c|}    
\hline
& \   \ 
&\ \ a \ \ &\ \ b \ \   \\
\hline  
\rightarrow & \  0 \
&   \{1,2\}  &     \\
\hline  
              & 1
&    \{1\}   & \{ 2 \}\\
\hline  
\leftarrow& 2
&   \{2\}   &  \{1\} \\
\hline  
\end{array}
$
\end{center}
}
\end{minipage}
\hspace{.3cm}
\begin{minipage}[b]{0.45\linewidth}
\caption{NFA $\cN^R$.}
\label{tab:Nr}
{\footnotesize
\begin{center}
$
\begin{array}{|c| c||c| c| }    
\hline
& \ \  \ \ 
&\ \ a \ \ &\ \ b \ \ \\
\hline  
\leftarrow & 0
&   &  \\
\hline  
			& \ 1 \
& \{0,1\} & \{2\} \\
\hline  
\rightarrow & 2
 &  \{0,2\} &  \{1\}  \\
\hline  
\end{array}
$
\end{center}}
\end{minipage}
\end{table}

\begin{table}[hbt]
\begin{minipage}[b]{0.4\linewidth}
\caption{DFA $\cN^{RD}$.}
\label{tab:Nrd}
{\footnotesize
\begin{center}
$
\begin{array}{|c| c||c| c| c|c|}    
\hline
& 
&\ \ a \ \ &\ \ b \ \   \\
\hline  
\rightarrow & \{2\}
&   \{0,2\}  &  \{1\}   \\
\hline  
\leftarrow & \{0,2\}
&    \{0,2\}   & \{ 1 \}\\
\hline  
		& \{1\}
&   \{0,1\}   &  \{2\} \\
\hline  
\leftarrow& \{0,1\}
&   \{0,1\}   &  \{2\} \\
\hline  
\end{array}
$
\end{center}}
\end{minipage}
\hspace{.5cm}
\begin{minipage}[b]{0.4\linewidth}
\caption{NFA $\cN^{RDR}$.}
\label{tab:Nrdr}
{\footnotesize
\begin{center}
$
\begin{array}{|c| c||c| c| c|c|}    
\hline
& 
&\ \ a \ \ &\ \ b \ \   \\
\hline  
\leftarrow & \{2\}
&     & \{ \{1\} , \{0,1\} \} \\
\hline  
\rightarrow & \{0,2\}
&   \{ \{2\},  \{0,2\} \}  & \\
\hline  
		& \{1\}
&     & \{ \{2\},  \{0,2\} \} \\
\hline  
\rightarrow& \{0,1\}
&   \{ \{1\} , \{0,1\} \}   &   \\
\hline  
\end{array}
$
\end{center}}
\end{minipage}
\end{table}
In summary, in this section we present the following contributions:
\bi
\item
We define partial \'atomata and study their properties.
\item
We prove that partial \'atomata are isomorphic to Sengoku's disjoint NFAs, obtained from any NFA $\cN$ by finding $\cN^{RDR}$.
\item
We prepare the ground for definitions of \'atomata in the next section.
\'Atomata are special cases of partial \'atomata and have even nicer properties than partial \'atomata.
\ei

\subsection{Partial \'Atomata}
Let $L$ be a non-empty regular language and 
let $\cN=(Q, \Sig, \delta, I,F)$ be any NFA accepting $L$, with
state set $Q=\{q_0,\ldots,q_{k-1}\}$.
Let   $L_i=L_{q_i,F}(\cN)$, $i\in\{0,\ldots,k-1\}$ be 
the  \emph{partial quotients} of $\cN$.
A partial quotient $L_i$ is \emph{initial} if $q_i$ is an initial state of $\cN$,
it is \emph{final} if $q_i$ is a final state.

A \emph{partial atom} of $\cN$ is any non-empty language 
of the form 
$\widetilde{L_0}\cap\widetilde{L_1}\cap \cdots \cap \widetilde{L_{k-1}}$, 
where $\widetilde{L_i}$ is either $L_i$ or $\ol{L_i}$.
A partial atom is \emph{initial} if it has some initial partial quotient 
$L_i$ as a term in its intersection, 
it is \emph{final} if and only if it contains $\eps$.
Since $L(\cN)$ is non-empty, $\cN$ has at least one partial quotient containing $\eps$. 
Hence it has exactly one final partial atom
$\widehat{L_0}\cap\widehat{L_1}\cap \cdots \cap \widehat{L_{k-1}}$, where 
$\widehat{L_i}=L_i$ if $\eps\in L_i$, and $\widehat{L_i}=\ol{L_i}$ otherwise.

If the intersection 
$\ol{L_0}\cap \cdots \cap \ol{L_{k-1}}$ is non-empty, then we call it 
the \emph{negative} partial atom; all the other partial atoms are \emph{positive}. 
Let the set of partial atoms be $X=\{X_0,\ldots,X_{\ell -1}\}$. 
Let the number of positive partial atoms be $h$; this number is either $\ell$ or $\ell-1$.
By convention, $I_X$ is the set of initial partial atoms,
$X_{h -1}$ is the final partial atom, and $X_{\ell-1}$ is the negative partial atom, if  present.
The negative partial atom can never be final, 
since there must be at least one complemented final partial quotient in its intersection.

In the following definition we use a one-one correspondence 
$X_i \leftrightarrow  {\mathbf X}_i$ between partial atoms $X_i$ and the states 
${\mathbf X}_i$ of the NFA $\cX$ defined below.

\begin{definition}
\label{def:partial_atomaton}
The \emph{partial \'atomaton} of  $\cN$, is the NFA defined by 
$\cX=({\mathbf X},\Sig,\eta, {\mathbf I_X},\{{\mathbf X}_{h -1}\})$,
 where ${\mathbf X}=\{{\mathbf X}_i\mid X_i\in X\}$,
 ${\mathbf I_X}=\{{\mathbf X}_i\mid X_i\in I_X\}$, 
 and ${\mathbf X}_j \in \eta({\mathbf X}_i, a)$ if and only if 
$aX_j \subseteq X_i$, 
for all ${\mathbf X_i},{\mathbf X_j}\in {\mathbf X}$ and $a\in\Sig$.
\end{definition}

The partial \'atomaton can be constructed directly from the NSE corresponding to $\cN$, as illustrated in Example~\ref{ex:partial}.

\begin{proposition}
\label{prop:quotient_partial}
The following properties hold for partial atoms:\\
1. Partial atoms are pairwise disjoint, that is, $X_i\cap X_j=\emp$ for all 
$i,j\in\{0,\ldots,\ell-1\}$, $i\neq j$.\\
2. The quotient $w^{-1}L$ of $L$ by $w\in\Sig^*$ is a (possibly empty) union of 
partial atoms.\\
3. The quotient $w^{-1}X_i$ of $X_i$ by $w\in \Sig^*$ is 
a (possibly empty) union of partial atoms.\\
4. Partial atoms define a partition of $\Sig^*$.
\end{proposition}
\begin{proof}
1. If $X_i\neq X_j$, then there exists $g\in\{0,\ldots,k-1\}$ such that 
$L_g$ is a term of $X_i$ and $\ol{L_g}$ is a term of $X_j$, or \emph{vice versa}. 
Hence $X_i\cap X_j=\emp$.\\
\hglue10pt 
2. The empty quotient, if present, is the empty union of partial atoms. 
Every non-empty quotient $K_j$ is a union of some partial quotients. 
As well, every $L_i\neq\emp$ is the union of all the $2^{\ell-1}$ 
intersections that have $L_i$ as a term. 
This includes all partial atoms that have $L_i$ as a term, and possibly 
some empty intersections.\\
\hglue10pt 
3. The quotient of a partial atom $X_i$ by a letter $a\in\Sig$ is an intersection 
of quotients of uncomplemented or complemented partial quotients of $L$. 
Since a quotient of a partial quotient is a union of partial quotients, 
and a quotient of a complemented partial quotient is an intersection of complemented 
partial quotients, the quotient of $X_i$ by $a$ is a union of intersections of
complemented or uncomplemented partial quotients of $L$.
If a partial quotient $L_j$ does not appear as a term in some intersection $Z$ 
of this union, then we ``add it in'' by using the fact that 
$Z=Z\cap(L_j\cup \ol{L_j})=(Z\cap L_j) \cup (Z\cap \ol{L_j})$.
After all the missing partial quotients are so added, we obtain a union of 
partial atoms.
It  follows that $w^{-1}X_i$ is a union of partial atoms of $L$ for every $w\in\Sig^*$.\\
\hglue10pt 
4. Since the union of all the intersections of complemented and uncomplemented partial atoms is $\Sig^*$, the claim follows.
\end{proof}

\begin{lemma}
\label{lem:inclusion_partial}
Let $w,x\in\Sig^*$.
If $wx\in X_i$ and $x\in X_j$ then $wX_j\subseteq X_i$, for $i,j\in\{0,\ldots,\ell-1\}$.
\end{lemma}
\begin{proof}
Assume that $wx\in X_i$ and $x\in X_j$, but suppose $wy\not\in X_i$ for some  
$y\in X_j$. 
Then $x\in w^{-1}X_i$ and $y\not\in w^{-1}X_i$.
By Proposition~\ref{prop:quotient_partial}, Part 3, $w^{-1}X_i$ is a union of 
partial atoms. So, on the one hand, $x\in w^{-1}X_i$ and $x\in X_j$ together imply 
$X_j\subseteq w^{-1}X_i$. On the other hand, from $y\not\in w^{-1}X_i$ and $y\in X_j$, 
we get $X_j\not\subseteq w^{-1}X_i$. So if $wy\not\in X_i$, we have a contradiction.
Hence, $wX_j\subseteq X_i$. 
\end{proof}

\begin{lemma}
\label{lem:path_partial}
For $w\in \Sig^*$, 
${\mathbf X}_j \in \eta({\mathbf X}_i, w)$ if and only if 
$wX_j \subseteq X_i$, for $i,j\in\{0,\ldots,\ell-1\}$.
\end{lemma}
\begin{proof}  
The proof is by induction on the length of $w$. 
If $|w|=0$ and ${\mathbf X}_j \in \eta({\mathbf X}_i, \eps)$, then $i=j$ and 
$\eps X_j \subseteq X_i$. If $|w|=0$ and $\eps X_j \subseteq X_i$, then $i=j$, 
since partial atoms are disjoint; hence ${\mathbf X}_j \in \eta({\mathbf X}_i, \eps)$.
If $|w|=1$, then the lemma holds by Definition~\ref{def:partial_atomaton}.  

Now, let $w=av$, where $a\in\Sig$ and $v\in\Sig^+$, and assume that 
lemma holds for $v$. 
Suppose that ${\mathbf X}_j \in \eta({\mathbf X}_i, av)$.
Then there exists some state ${\mathbf X}_k$ such that
${\mathbf X}_k \in \eta({\mathbf X}_i, a)$  and 
${\mathbf X}_j \in \eta({\mathbf X}_k, v)$.
Thus,  $aX_k \subseteq X_i$ by the definition of partial \'atomaton, 
and $vX_j \subseteq X_k$ by the induction assumption, implying that
$avX_j \subseteq X_i$.
 
Conversely, let $avX_j \subseteq X_i$. Then $vX_j \subseteq a^{-1}X_i$.
Let $x\in X_j$. Then $vx\in a^{-1}X_i$. 
Since by Proposition~\ref{prop:quotient_partial}, Part 3, $a^{-1}X_i$ is a union 
of partial atoms, there exists a partial atom $X_k$ such that $vx\in X_k$.
Since $x\in X_j$,  by Lemma~\ref{lem:inclusion_partial} we get 
$vX_j \subseteq X_k$.
Furthermore, because $avX_j \subseteq X_i$ and $x\in X_j$, we have 
$avx \in X_i$. Since $vx\in X_k$, then $aX_k \subseteq X_i$ by 
Lemma~\ref{lem:inclusion_partial}.

As the lemma holds for $v$ and $a$, $vX_j \subseteq X_k$ implies
${\mathbf X}_j \in \eta({\mathbf X}_k, v)$, and 
$aX_k \subseteq X_i$ implies
${\mathbf X}_k \in \eta({\mathbf X}_i, a)$, showing that
${\mathbf X}_j \in \eta({\mathbf X}_i, av)$.
\end{proof}

\begin{proposition}
\label{prop:right_lang_partial}
The right language of state ${\mathbf X}_i$ of partial \'atomaton $\cX$ is 
the partial atom $X_i$, that is, 
$L_{{\mathbf X}_i,\{{\mathbf X}_{h-1}\}}(\cX)=X_i$, for all $i\in\{0,\ldots,\ell-1\}$.
\end{proposition}
\begin{proof}  
Let $w\in L_{{\mathbf X}_i,\{{\mathbf X}_{h-1}\}}(\cX)$; then 
${\mathbf X}_{h-1}\in \eta({\mathbf X}_i,w)$.
By Lemma~\ref{lem:path_partial}, we have $wX_{h-1} \subseteq X_i$.  
Since $\eps\in X_{h-1}$,  we have $w\in X_i$. 

Now suppose that $w\in X_i$. Then $w\eps\in X_i$, and since $\eps\in X_{h-1}$,
by Lemma~\ref{lem:inclusion_partial} we get $wX_{h-1}\subseteq X_i$. 
By Lemma~\ref{lem:path_partial}, 
${\mathbf X}_{h-1} \in\eta({\mathbf X}_i,w)$, that is, 
$w\in L_{{\mathbf X}_i,\{{\mathbf X}_{h-1}\}}(\cX)$. 
\end{proof}

\begin{proposition}
\label{prop:lang_partial_atomaton}
The language accepted by partial \'atomaton $\cX$ of $L$ is $L$,
that is, $L(\cX)=L$.
\end{proposition}
\begin{proof}
We have $L(\cX)=\bigcup_{X_i\in I_X} L_{{\mathbf X}_i,\{{\mathbf X}_{h-1}\}}(\cX)= 
\bigcup_{X_i\in I_X}X_i$, by Proposition~\ref{prop:right_lang_partial}.
Since $I_X$ is the set of all partial atoms that have some $L_j$ as a term 
such that $q_j\in I$, we also have $L=\bigcup_{X_i\in I_X}X_i$. 
\end{proof}

Next, we will show that $\cX$ is isomorphic to the NFA $\cN^{\rev\deter\rev}$. 
To prove this result, we use the automata 
$\cN^\rev=(Q,\Sig,\delta^\rev, F,I)$,
$\cN^{\rev\deter}=(S,\Sig,\gamma, F, G)$, and
$\cN^{\rev\deter\rev}=(S,\Sig,\gamma^\rev, G, \{F\})$.
This is a generalization of the isomorphism result in~\cite{BrTa13}.

\begin{proposition}[Isomorphism]
\label{prop:isomorphism}
Let $\varphi: {\mathbf X} \to {S}$ be the mapping assigning to state 
${\mathbf X}_i$, given by
$X_i=L_{i_0}\cap\cdots\cap L_{i_{g-1}}\cap\ol{L_{i_{g}}}
\cap\cdots\cap \ol{L_{i_{k-1}}}$ of $\cX$, the set 
$\{q_{i_0},\ldots, q_{i_{g-1}}\}$.
Then $\varphi$ is an NFA isomorphism between $\cX$ and $\cN^{\rev\deter\rev}$. 
\end{proposition}
\begin{proof}
Every initial state ${\mathbf X}_i$ of $\cX$ is mapped to a subset $Q_i$ of $Q$, 
corresponding to the set of uncomplemented $L_j$'s in $X_i$, 
having a property $Q_i\cap I\ne\emp$. 
Then $Q_i$ is a final state of $\cN^{\rev\deter}$ and therefore, 
an initial state of $\cN^{\rev\deter\rev}$.

The final state ${\mathbf X}_{h -1}$ of $\cX$ is mapped to 
the set of all $q_j$'s such that $\eps\in L_j$, that is, 
the set of final states of $\cN$, 
which is the initial state of $\cN^{\rev\deter}$, 
and thus the final state of $\cN^{\rev\deter\rev}$.

We also have to demonstrate that ${\mathbf X}_j\in\eta({\mathbf X}_i,a)$ if and 
only if $\varphi({\mathbf X}_j)\in\gamma^\rev(\varphi({\mathbf X}_i),a)$
for all ${\mathbf X}_i, {\mathbf X}_j\in {\mathbf X}$ 
and $a\in\Sig$.

Let $s\in S$ be a state of $\cN^{\rev\deter}$.
The left language of state $s$ consists of all words $u$ such that 
$u\in L_{F,q}(\cN^\rev)$ for every $q\in s$, 
but $u\not\in L_{F,q'}(\cN^\rev)$ for any $q'\not\in s$.
We get that, $L_{F,s}(\cN^{\rev\deter})=
(\bigcap_{q\in s} L_{F,q}(\cN^\rev)) \setminus 
(\bigcup_{q'\not\in s} L_{F,q'}(\cN^\rev))=
(\bigcap_{q\in s} L_{F,q}(\cN^\rev)) \cap 
(\bigcap_{q'\not\in s} \ol{L_{F,q'}(\cN^\rev)})$, and so
$L_{s,F}(\cN^{\rev\deter\rev})=
(\bigcap_{q\in s} L_{q,F}(\cN)) \cap 
(\bigcap_{q'\not\in s} \ol{L_{q',F}(\cN)})$.
Also, given a state $t$ of $\cN^{\rev\deter}$ (as well as $\cN^{\rev\deter\rev}$), 
similar equations hold for $t$. 
Then, $\gamma(s,a)=t$ for some $a\in\Sig$ if and only if
$L_{F,s}(\cN^{\rev\deter})a\subseteq L_{F,t}(\cN^{\rev\deter})$.
This is equivalent to having $s\in\gamma^\rev(t,a)$ if and only if
$aL_{s,F}(\cN^{\rev\deter\rev})\subseteq L_{t,F}(\cN^{\rev\deter\rev})$.
Considering above, the latter is equivalent to
$a(\bigcap_{q\in s} L_{q,F}(\cN) \cap \bigcap_{q'\not\in s} \ol{L_{q',F}(\cN)})
\subseteq \bigcap_{q\in t} L_{q,F}(\cN) \cap \bigcap_{q'\not\in t} \ol{L_{q',F}(\cN)}$.

Let $s=\{q_{j_0},\ldots, q_{j_{e-1}}\}$ and $t=\{q_{i_0},\ldots, q_{i_{g-1}}\}$.
Then we have that 
$\{q_{j_0},\ldots,q_{j_{e-1}}\}\in\gamma^\rev(\{q_{i_0},\ldots, q_{i_{g-1}}\},a)$ 
if and only if
$a(L_{j_0}\cap\cdots\cap L_{j_{e-1}}\cap\ol{L_{j_{e}}}\cap\cdots\cap \ol{L_{j_{k-1}}}) 
\subseteq L_{i_0}\cap\cdots\cap L_{i_{g-1}}\cap\ol{L_{i_{g}}}\cap\cdots\cap \ol{L_{i_{k-1}}}$.
By denoting $X_i=L_{i_0}\cap\cdots\cap L_{i_{g-1}}\cap\ol{L_{i_{g}}}\cap\cdots\cap \ol{L_{i_{k-1}}}$ 
and $X_j=L_{j_0}\cap\cdots\cap L_{j_{e-1}}\cap\ol{L_{j_{e}}}\cap\cdots\cap \ol{L_{j_{k-1}}}$, 
we get that $\varphi({\mathbf X}_j)\in\gamma^\rev(\varphi({\mathbf X}_i),a)$ 
if and only if $aX_j \subseteq X_i$.
According to the definition of $\cX$, the latter is equivalent to 
${\mathbf X}_j \in \eta({\mathbf X}_i, a)$.
\end{proof}

\begin{corollary}
\label{cor:nfa_isomorphism}
The mapping $\varphi$ is a DFA isomorphism between $\cX^\rev$ and 
$\cN^{\rev\deter}$.
\end{corollary}

\section{Quotients, Atoms and \'Atomata}
\label{sec:atoms}

\subsection{Background}
Sengoku~\cite{Sen92} studied the NFA $\cN^{RDMR}$ obtained from any NFA $\cN$ by reversal, determinization, minimization and  reversal.  He called $\cN^{RDMR}$ the \emph{normal} NFA equivalent to $\cN$. He defined an NFA $\cN$ to be in \emph{standard form}
if $\cN^{RD}$ is minimal. 

Recall a (slightly modified version of a) theorem from~\cite{Brz63}:
\begin{theorem}
\label{thm:Brz}
If  an NFA $\cN$ has no empty states and $\cN^\rev$ is deterministic, 
then $\cN^\deter$ is minimal.
\end{theorem}

Suppose instead of starting with an NFA, we start with a minimal DFA $\cD$. 
Since $\cD$ has no unreachable states, $\cD^R$ has no empty states, and 
so Theorem~\ref{thm:Brz} applies to $\cD^R$. Thus 
$\cD^{RD}$ is minimal, that is, $\cD^{RD}=\cD^{RDM}$, and $\cD^{RDMR}=\cD^{RDR}$.
Since $\cD^{RD}$ is minimal, $\cD^{RDR}$ is in standard form.
The NFA $\cD^{RDR}$ also appeared in the work of Matz and Potthoff~\cite{MaPo95}.

As in the case of disjoint NFAs discussed in the previous section, we introduce a completely different definition of an NFA (which we call an \'atomaton) defined by a given minimal DFA---or equivalently, by any regular language---and prove that that NFA is isomorphic to $\cD^{RDR}$.

The concepts used here are special cases of those of Section~\ref{sec:partial}: 
here, instead of using an arbitrary NFA, we start with the minimal DFA of a regular language $L$.

\subsection{\'Atomata}

Let $\cD=(Q, \Sig, \delta, q_0,F)$ be the minimal DFA of $L$, with
state set $Q=\{q_0,\ldots,q_{n-1}\}$.
It is well known that the right language of every state $q_i$ of $\cD$ 
is a quotient $K_i=L_{q_i,F}(\cD)$ of $L$, $i\in\{0,\ldots,n-1\}$.

An \emph{atom} of $L$ is any non-empty language of the form 
$\widetilde{K_0}\cap\widetilde{K_1}\cap \cdots \cap \widetilde{K_{n-1}}$, 
where $\widetilde{K_i}$ is either $K_i$ or $\ol{K_i}$. 
Let the set of atoms be $A=\{A_0,\ldots,A_{m-1}\}$. 
Thus atoms of $L$ define a partition of $\Sig^*$, and 
$L$ has at most $2^n$ atoms.

An atom is \emph{initial} if it has $K_0$ (rather than $\ol{K_0}$) as a term;
it is \emph{final} if and only if it contains $\eps$.
Since $L$ is non-empty, it has at least one quotient containing $\eps$. 
Hence it has exactly one final atom, the atom 
$\widehat{K_0}\cap\widehat{K_1}\cap \cdots \cap \widehat{K_{n-1}}$, where 
$\widehat{K_i}=K_i$ if $\eps\in K_i$, and $\widehat{K_i}=\ol{K_i}$ otherwise.

If the intersection 
$\ol{K_0}\cap \cdots \cap \ol{K_{n-1}}$ is non-empty, then we call it 
the \emph{negative} atom; all the other atoms are \emph{positive}. 
Let the number of positive atoms be $p$; this number is either $m$ or $m-1$.
By convention, $I_A$ is the set of initial atoms, $A_{p-1}$ is the final atom, 
and the negative atom, if present, is $A_{m-1}$.
The negative atom can never be final, 
since there must be at least one complemented final quotient in its intersection.

Evidently, the set of partial atoms of the quotient DFA $\cD$ of the language $L$ is 
the set of atoms of $L$.
Since atoms of $L$ are a special case of partial atoms, all the results of 
Section~\ref{sec:partial} about partial atoms hold for atoms.

\vskip1em

Let $\cN=(Q, \Sig, \delta, I,F)$ be an NFA accepting $L$, with 
partial quotients $L_0,\ldots,L_{k-1}$,
partial atoms $X=\{X_0,\ldots,X_{\ell -1}\}$, and 
partial \'atomaton 
$\cX=({\mathbf X},\Sig,\eta, {\mathbf I_X},\{{\mathbf X}_{h -1}\})$.

\begin{proposition}
\label{prop:atomsubset}
For every $X_i$, where $i=0,\ldots,\ell -1$, there exists some atom $A_j$, 
$j\in \{0,\ldots,m-1\}$, such that $X_i\subseteq A_j$. 
\end{proposition}
\begin{proof}
Let $X_i=L_{i_0}\cap\cdots\cap L_{i_{g-1}}\cap\ol{L_{i_{g}}}
\cap\cdots\cap \ol{L_{i_{k-1}}}$, where $0\le g\le k$.
Let $A_j=K_{j_0}\cap\cdots\cap K_{j_{e-1}}\cap\ol{K_{j_{e}}}
\cap\cdots\cap \ol{K_{j_{n-1}}}$ be an atom that has a quotient $K_l$ 
uncomplemented if and only if there is some $L_r\in\{L_{i_0},\ldots, L_{i_{g-1}}\}$ 
such that $L_r\subseteq K_l$, and all the other quotients complemented.
We claim that $X_i\subseteq A_j$.
On the one hand, from the choice of atom $A_j$, it is clear that 
$L_{i_0}\cap\cdots\cap L_{i_{g-1}}\subseteq K_{j_0}\cap\cdots\cap K_{j_{e-1}}$.
On the other hand, it has to be the case that every quotient $K_l$ that is 
complemented in $A_j$, is a union of some $L_h$'s from the set 
$\{L_{i_g},\ldots, L_{i_{k-1}}\}$, or otherwise $K_l$ would be included
as an uncomplemented quotient.
Therefore, $K_{j_e}\cup\cdots\cup K_{j_{n-1}}\subseteq L_{i_g}\cup\cdots\cup L_{i_{k-1}}$,
implying
$\ol{L_{i_{g}}}\cap\cdots\cap\ol{L_{i_{k-1}}}\subseteq\ol{K_{j_e}}\cap\cdots\cap \ol{K_{j_{n-1}}}$.
It follows that $L_{i_0}\cap\cdots\cap L_{i_{g-1}}\cap\ol{L_{i_{g}}}
\cap\cdots\cap \ol{L_{i_{k-1}}}\subseteq
K_{j_0}\cap\cdots\cap K_{j_{e-1}}\cap\ol{K_{j_{e}}}\cap\cdots\cap \ol{K_{j_{n-1}}}$.
Thus, $X_i\subseteq A_j$.
\end{proof}

\begin{proposition}
\label{prop:union}
Every atom $A_j$ is a disjoint union of some $X_i$s.
\end{proposition}
\begin{proof}
The set $X$, as well 
the set $A$ of atoms is a partition of $\Sig^*$.
By Proposition~\ref{prop:atomsubset}, every $X_i$ is a subset of
some atom $A_j$; hence we conclude that every atom is a disjoint union of 
some $X_i$s.  
\end{proof}

We define the \'atomaton of $L$ as a special case of a partial \'atomaton
that uses a one-one correspondence 
$A_i \leftrightarrow  {\mathbf A}_i$ between atoms $A_i$ and the states 
${\mathbf A}_i$ of the NFA $\cA$ as follows:

\begin{definition}
\label{def:atomaton}
The \emph{\'atomaton} of $L$
 is the NFA $\cA=({\mathbf A},\Sig,\alpha, {\mathbf I}_A,\{{\mathbf A}_{p-1}\}),$
 where ${\mathbf A}=\{{\mathbf A}_i\mid A_i\in A\}$,
 ${\mathbf I}_A=\{{\mathbf A}_i\mid A_i\in I_A\}$, 
 and ${\mathbf A}_j \in \alpha({\mathbf A}_i, a)$ if and only if 
$aA_j \subseteq A_i$, for all ${\mathbf A_i},{\mathbf A_j}\in {\mathbf A}$ and 
$a\in\Sig$.
\end{definition}

\begin{proposition}
\label{prop:partial_atomaton}
Suppose $\cN$ is an NFA accepting $L$ and $\cX$ is its partial \'atomaton; 
then $\cX$  is the \'atomaton of $L$ if and only if $X$ is the set of atoms.
\end{proposition}
\begin{proof}
If $\cX$ is the \'atomaton of $L$, then the set $X$ must be the set of atoms.

Conversely, let $X$ be the set of atoms. 
We show that in this case, $I_X$ is the set of initial atoms, and 
$X_{h-1}$ is the final atom.
By definition of $I_X$, $X_i\in I_X$ if and only if $X_i$ has some term 
$L_j$ such that $L_j=L_{q_j,F}(\cN)$ for some initial state $q_j$ of $\cN$. 
So if $X_i\in I_X$, then there is some $q_j\in I$
such that $X_i\subseteq L_j\subseteq K_0$ holds, implying 
that $X_i$ is an initial atom.
On the other hand, if $X_i$ is an initial atom, then it must have some term
$L_j$ such that $X_i\subseteq L_j\subseteq K_0$ and $q_j\in I$, 
implying $X_i\in I_X$.

Also, since $\eps\in X_{h -1}$, $X_{h -1}$ is the final atom. 
One can verify now that the partial \'atomaton $\cX$ of 
Definition~\ref{def:partial_atomaton} is the \'atomaton of $L$. 
\end{proof}

We illustrate the computation of the \'atomaton using quotient equations.

\begin{example}
\label{ex:atoms1}
Consider the language $L=a\Sig^*$ of Example~\ref{ex:partial}.
It is defined by the following quotient equations:
\begin{displaymath}
 \begin{array}{rl}
   K_0 &= aK_1 \cup b K_2=L, \\
   K_1 &= aK_1\cup bK_1 \cup \eps=\Sig^*, \\
   K_2 &= aK_2\cup bK_2=\emp. 
 \end{array}  
\end{displaymath}  
We find the atoms using these quotient equations in the same way 
as we found partial atoms from the equations for partial quotients. 
Note that all intersections having $K_2$ as a term are empty, 
as are those containing $\ol{K_1}$.
Hence there are only two atoms: 
$A_0=K_0\cap K_1\cap\ol{K_2}=K_0=L$, and
$A_1=\ol{K_0}\cap K_1\cap\ol{K_2}=\ol{K_0}=\ol{L}$. 
Thus we find the atom equations
\begin{displaymath}
 \begin{array}{rl}
  A_0
 &= a[(K_0\cap K_1\cap \ol{K_2}) \cup (\ol{K_0}\cap K_1\cap \ol{K_2})],\\
  A_1&= 
  b[(K_0\cap K_1\cap \ol{K_2}) \cup (\ol{K_0}\cap K_1\cap \ol{K_2})]\cup \eps,
 \end{array}
\end{displaymath}
where $A_0=L=a\Sig^*$ and $A_1=\ol{L}=b\Sig^*\cup \eps$.
By Proposition~\ref{prop:union}, every atom is a union of partial atoms. 
Indeed one verifies that $A_0=X_0\cup X_1$ and  $A_1=X_2\cup X_3$, 
where the $X_i$ are defined in Example~\ref{ex:partial}.
\qedb
\end{example}

We now relate a number of concepts associated with regular languages:

\begin{theorem}
\label{thm:atomaton}
Let $L$ be a regular language, let $\cD$ be its minimal DFA, 
and let $\cA$ be its \'atomaton. Then the following statements hold:\\
1. $\cA$ is isomorphic to $\cD^{\rev\deter\rev}$.\\
2. The reverse $\cA^\rev$ of $\cA$ is the minimal DFA of $L^R$.\\
3. The DFA $\cA^\deter$ is the minimal DFA of $L$.\\
4. For any NFA $\cN$ accepting $L$, $\cN^{\rev\deter\mini\rev}$ is isomorphic to  $\cA$.\\
5. $\cA$ is isomorphic to $\cD$ if and only if $L$ is bideterministic.\\
\end{theorem}
\begin{proof}
Claim 1 follows by Propositions~\ref{prop:isomorphism} and~\ref{prop:partial_atomaton}.
Claim 2 follows from Claim 1 and  Theorem~\ref{thm:Brz}.
Since $\cA^R$ is deterministic and minimal, it has no unreachable states. 
Hence $\cA$ has no empty states and Theorem~\ref{thm:Brz} applies. 
Therefore $\cA^D$ is the minimal DFA accepting $L$, and Claim 3 follows. 
Claim 4 holds because $\cN^{RDM}$ is the minimal DFA of $L^R$.

To prove Claim 5, first suppose that $\cA$ is isomorphic to $\cD$.
DFA $\cD$ must be trim, because all states of \'atomaton $\cA$ are non-empty.  
Since  $\cA$ is isomorphic to $\cD$, $\cA$ itself is a trim DFA. 
By Claim~2, $\cA^\rev$ is a DFA. 
Hence $\cA$, and so also $L$, are bideterministic.

Conversely, let $\cB$ be a bideterministic DFA accepting $L$. 
Since $\cB$ is a trim DFA,  $\cB^{\rev\deter}$ is minimal by Theorem~\ref{thm:Brz}.
Since $\cB^\rev$ is deterministic, we get $\cB^{\rev\deter}=\cB^{\rev}$.
Thus $\cB^{\rev\deter\mini\rev}=\cB^{\rev\deter\rev}=\cB^{\rev\rev}=\cB$
is isomorphic to $\cA$ by Claim~4.
On the other hand, 
since $\cB^{\rev}$ is deterministic, $\cB^{\deter}=\cB$ is minimal 
by Theorem~\ref{thm:Brz}. Hence $\cB$ is isomorphic to $\cD$.
Since $\cB$ is isomorphic both to $\cA$ and $\cD$, 
$\cA$ is isomorphic to $\cD$.
\end{proof}

An NFA $\cD^{\rev\deter\rev\trim}$ isomorphic to the trim \'atomaton $\cA^{\trim}$
is considered in~\cite{MaPo95}. It is noted there that for each word $w$ in $L$ 
there is a unique path in $\cD^{\rev\deter\rev\trim}$ accepting $w$, 
and deleting any transition from $\cD^{\rev\deter\rev\trim}$ results 
in a smaller accepted language. 
It is also stated in~\cite{MaPo95} without proof that the
right language $L_{q,F}(\cN)$ of any state $q$ of an NFA $\cN$ accepting $L$ is 
a subset of a union of atoms. This holds because $L_{q,F}(\cN)$ is a subset 
of some quotient of $L$, and quotients are unions of atoms by 
Proposition~\ref{prop:quotient_partial}, Part 2.

Theorem~\ref{thm:atomaton} provides another method of finding the \'atomaton 
of $L$: simply reverse the quotient DFA of $L^R$.

\vskip1em
 
To end this section, we explain the differences between our present  
definition of an atom and that of~\cite{BrTa11}.
The definition  in \cite{BrTa11} did not consider 
the intersection of all the complemented quotients to be an atom, and
so all atoms were positive.
It was shown in~\cite{BrTa11} that the reverse of the \'atomaton
with only positive atoms is the trim version of the minimal DFA of $L^R$. 
With the negative atom, we avoid the trimming operation;
so the reverse of the \'atomaton is the minimal DFA of $L^R$. 
Also, with the negative atom, a language $L$ and its complement language $\ol{L}$ 
have the same atoms. Finally, we have symmetry between the atoms with 
0 and $n$ complemented quotients, and the same upper bounds on quotient 
complexity for both, as was shown in \cite{BrTa13}.

\section{Atomic NFAs}
\label{sec:atomic}

\subsection{Basic Properties}
We now introduce a new class of NFAs and study their properties.

\begin{definition}
\label{def:atomic2}
An NFA $\cN=(Q,\Sigma,\delta,I,F)$ is \emph{atomic} if for every state 
$q\in Q$, the right language $L_{q,F}(\cN)$ of $q$ is a union of some atoms of $L(\cN)$. 
\end{definition}
Note that, if $L_{q,F}(\cN)=\emp$, then it is the union of zero atoms.

Recall that an NFA $\cN$ is residual, if $L_{q,F}(\cN)$ is a (left) quotient of 
$L(\cN)$ for every $q\in Q$.
Since every quotient is a union of atoms (see Proposition~\ref{prop:quotient_partial}, 
Part~2), every residual NFA is atomic. 
However, the converse is not true: there exist atomic NFAs which are not residual.
For example, the \'atomaton of a language $L$ is atomic, but not necessarily residual, because
in a general case, atoms are different from quotients. 
Note also that every DFA with only reachable states is atomic because
the right language of every state of such DFA is some quotient. 

Let us now consider the universal automaton $\cU_L=(Q,\Sigma,\delta,I,F)$ of
a language $L$. We state some basic properties of this automaton from~\cite{LoSa07}.
Let $(X,Y)$ be a factorization of $L$.
Then \\
\hglue 10pt (1)
$Y=\bigcap_{x\in X} x^{-1}L$ 
and $X=\bigcap_{y\in Y} Ly^{-1}$.\\
\hglue 10pt (2)
$L_{I,(X,Y)}(\cU_L)=X$ and 
$L_{(X,Y),F}(\cU_L)=Y$.\\
\hglue 10pt (3) The universal automaton $\cU_L$ accepts $L$.

\begin{theorem}
\label{thm:univ_atomic}
Let $L$ be any regular language. The following automata accepting $L$ 
are atomic:
\newline
1. The \'atomaton $\cA$.
\newline
2. Any DFA with no unreachable states.
\newline
3. Any residual NFA. 
\newline
4. The universal automaton $\cU_L$.
\end{theorem}
\vspace{-.3cm}
\begin{proof} 
1. $\cA$ is atomic because the right language of every state of $\cA$ is an atom of $L$.\\
\hglue10pt 
2. The right language of every state of any DFA accepting $L$ 
that has no unreachable states, 
is a quotient of~$L$. Since every quotient is a union of atoms, every such DFA 
is atomic.\\
\hglue10pt 
3. The right language of every state of any residual NFA of $L$ 
is a quotient of~$L$, and hence a union of atoms. Thus, any residual NFA is 
atomic.\\
\hglue10pt 
4. We show that the right language of every state $(X,Y)$ of $\cU_L$ is a union 
of atoms of $L(\cU_L)=L$. Let $(X,Y)$ be any state of $\cU_L$.
Since by property (2) above, $L_{(X,Y),F}(\cU_L)=Y$ holds, it is enough to show 
that $Y$ is a union of atoms. 

By (1), $Y=\bigcap_{x\in X} x^{-1}L$. 
We note that if $Y=\emptyset$, then $Y$ is the union of zero atoms.
We also note that if $X=\emptyset$, then $Y=\Sig^*$, and so $Y$ is the union of 
all atoms.

Let $L_0,\ldots, L_{n-1}$ be the quotients of $L$. 
Then for some $H\subseteq [n]$, $Y=\bigcap_{i\in H} L_i$.
Now $\bigcap_{i\in H} L_i=(\bigcap_{i\in H} L_i)\cap
(\bigcap_{j\in [n]\setminus H} (L_j\cup\ol{L_j})=
\bigcup (\bigcap_{i\in H} L_i)\cap(\bigcap_{j\in[n]\setminus H}
\widetilde{L_j})$, where $\widetilde{L_j}$ is either $L_j$ or $\ol{L_j}$. 
Thus, $Y$ is a union of atoms of $L$. 
\end{proof}

\subsection{Atomicity of States and NFAs}
\label{sec:detecting}

Let $\cN=(Q, \Sig, \delta, I,F)$ be any NFA accepting $L$.
We call a state $q_i$ of $\cN$ \emph{atomic} if its right language
$L_i=L_{q_i,F}(\cN)$ is a union of atoms of $L$. 
We now present a method of detecting which states of an NFA
are atomic.

Consider the DFA $\cN^{\rev\deter}$ and the NFA $\cN^{\rev\deter\rev}$;
these two automata have the same set $S$ of states. 
By Proposition~\ref{prop:isomorphism}, there is an isomorphism $\varphi$
between the partial \'atomaton $\cX$ of $\cN$ and the NFA $\cN^{\rev\deter\rev}$. 
Since there is a one-one correspondence between states ${\mathbf X}_i$ of 
$\cX$ and partial atoms $X_i$ of $\cN$, we can also establish a one-one 
correspondence $\varphi'$ between partial atoms $X_i$ of $\cN$ and 
states $s_i$ of $\cN^{\rev\deter}$ as follows:

\begin{definition}
\label{def:mapping}
Let $\varphi' : X \to S$ be the mapping such that for every $X_i\in X$ and 
$s_i\in S$, $\varphi'(X_i)= s_i$ if and only if $\varphi({\mathbf X}_i)= s_i$. 
\end{definition}

DFA $\cN^{\rev\deter}$ is not necessarily minimal. 
Let $S_0,\ldots,S_{r-1}$ be the sets of equivalent states 
of $\cN^{\rev\deter}$; that is, every $S_j$ is an equivalence class of 
the states of $\cN^{\rev\deter}$. 
The following proposition holds:

\begin{proposition}
\label{prop:atom_equiv}
Let $X'\subseteq X$ be a set of partial atoms of $\cN$, and 
let $S'\subseteq S$ be the corresponding set of states of $\cN^{\rev\deter}$
according to mapping $\varphi'$.
An equality $\bigcup_{X_i\in X'} X_i = A_j$ holds for some atom $A_j$ if and only if
$S'$ is equal to some equivalence class $S_j$.  
\end{proposition}
\begin{proof}
Let $X'\subseteq X$ be a set of partial atoms of $\cN$, and 
let $S'=\{s_i\in S \mid \varphi'(X_i)= s_i, \text{ where } X_i\in X'\}$ be 
the corresponding set of states of $\cN^{\rev\deter}$.

Consider the minimal DFA $\cN^{\rev\deter\mini}$ of $L^R$.  
It is well known that this DFA can be obtained 
by ``merging'' the states of each set $S_j$ of the states of the DFA
$\cN^{\rev\deter}$, into a state $t_j$ of $\cN^{\rev\deter\mini}$.
Similarly, $\cN^{\rev\deter\mini\rev}$, the reverse NFA of the minimal DFA of $L^R$,
can be obtained by merging the corresponding states of $\cN^{\rev\deter\rev}$, 
or equivalently, its isomorphic partial \'atomaton $\cX$.
Since by Proposition~\ref{prop:right_lang_partial},
the right language of every state of $\cX$ is some partial atom, 
the right language of the state $t_j$ of $\cN^{\rev\deter\mini\rev}$ is
the union of partial atoms $X_i$ of $\cN$ such that their corresponding states
$s_i=\varphi'(X_i)$ belong to $S_j$. 

On the other hand, according to Theorem~\ref{thm:atomaton}, Part 4, 
the NFA $\cN^{\rev\deter\mini\rev}$ is isomorphic to the \'atomaton of $L$, and 
by Proposition~\ref{prop:right_lang_partial},
the right language of every state of the \'atomaton is some atom.
We conclude that the union of partial atoms $X_i$ such that 
$s_i=\varphi'(X_i)$ belong to $S_j$, is some atom $A_j$. 
Since partial atoms are disjoint, no other union of partial atoms
of $\cN$ can be equal to $A_j$. Thus, the claim of the proposition holds.
\end{proof}

We use the equivalence classes $S_0,\ldots,S_{r-1}$ of the states of 
the DFA $\cN^{\rev\deter}$ to detect which states of $\cN$ are atomic.

\begin{theorem}
\label{thm:atomic_state}
A state $q_i$ of an NFA $\cN$ is atomic if and only if 
the subset $S'_i=\{s_j\in S \mid q_i\in s_j\}$ of states of 
$\cN^{\rev\deter}$ is a union of some equivalence classes of $\cN^{\rev\deter}$.
\end{theorem}
\begin{proof}
Consider a state $q_i$ of an NFA $\cN$ with right language $L_i$.
Let $X'$ be the set of partial atoms $X_j$ of $\cN$ 
that have $L_i$ uncomplemented in the intersection representing $X_j$, and
let ${\mathbf X}'$ be the corresponding set of states of the partial \'atomaton 
$\cX$ of $\cN$.
Let $S'_i$ be the set of states of $\cN^{\rev\deter\rev}$ that are assigned to
the states in ${\mathbf X}'$ by the mapping $\varphi$ of 
Proposition~\ref{prop:isomorphism}. Clearly, $S'_i$ consists of
exactly those states $s_j$ of $\cN^{\rev\deter\rev}$ such that $q_i\in s_j$. 

Now suppose that $L_i$ is a union of atoms.
Since $X'$ is the set of partial atoms of $\cN$ with $L_i$ uncomplemented, 
$L_i$ is equal to the union of all partial atoms in $X'$. 
So the union of all partial atoms in $X'$ is a union of atoms. 
By Definition~\ref{def:mapping}, partial atoms in $X'$ are 
mapped by $\varphi'$ exactly to the states in $S'_i$. 
By Proposition~\ref{prop:atom_equiv},
$S'_i$ is a union of some equivalence classes of $\cN^{\rev\deter}$. 

Conversely, if $L_i$ is not a union of atoms, then the union of 
partial atoms in $X'$ is not a union of atoms either. 
Contrarily to the case above, the set $S'_i$ cannot be a union of 
any equivalence classes of $\cN^{\rev\deter}$. 
\end{proof}

\begin{example}
\label{ex:atomic_state}
Consider the NFA $\cN$ of Table~\ref{tab:N} and the DFA $\cN^{RD}$ of Table~\ref{tab:Nrd}. 
The equivalence classes of the states of $\cN^{RD}$ are 
$S_0=\{ \{1\},\{2\} \} $ and $S_1=\{ \{0,1\},\{0,2\} \}$. 
Since $0$ appears in both states of $S_1$ and does not appear in the states of $S_0$, 
state $0$ of $\cN$ is atomic.
However, $1$ appears in the set $\{ \{1\},\{0,1\} \} $, which is not a union of 
equivalence classes; hence state $1$ of $\cN$ is not atomic. Similarly, state $2$ 
is not atomic.
\qedb
\end{example}

The following result is a consequence of Theorem~\ref{thm:atomic_state}:
\begin{corollary}
\label{cor:atomic}
An NFA $\cN$ is atomic if and only if $\cN^{\rev\deter}$ is minimal. 
\end{corollary}
\begin{proof}
If $\cN^{\rev\deter}$ is minimal, the equivalence classes of its states  are singletons. 
So the set of states of $\cN^{\rev\deter}$  in which a state $q_i$ of $\cN$ appears is a union of equivalence classes.
By Theorem~\ref{thm:atomic_state},
every state is atomic, and  so is $\cN$.

Conversely, suppose $\cN$ is atomic, but $\cN^{RD}$ is not minimal.
Then there are two states $s_j$ and $s_j'$ of $\cN^{\rev\deter}$  which are equivalent.
Without loss of generality, suppose that $q_i\in s_j\setminus s_j'$; 
then the set of states in which $q_i$ appears cannot be a union of equivalence classes. 
By Theorem~\ref{thm:atomic_state} again, $q_i$ is not atomic, and neither is $\cN$.
\end{proof}

We also have the following corollary:

\begin{corollary}
\label{cor:atomic2}
An NFA $\cN$ is atomic if and only if the partial atoms of $\cN$ are the atoms of $L$. 
\end{corollary}
\begin{proof}
By Corollary~\ref{cor:atomic}, an NFA $\cN$ is atomic if and only if 
$\cN^{\rev\deter}$ is minimal. 
But, if $\cN^{\rev\deter}$ is minimal, then $\cN^{RD}=\cN^{RDM}$ and 
$\cN^{\rev\deter\rev}=\cN^{\rev\deter\mini\rev}$.
By Theorem~\ref{thm:atomaton}, Part 4, 
$\cN^{\rev\deter\rev}$ is isomorphic to the \'atomaton $\cA$ of $L$. 
Since by Proposition~\ref{prop:isomorphism}, the partial \'atomaton $\cX$ 
of $\cN$ is isomorphic to $\cN^{\rev\deter\rev}$, $\cX$ and $\cA$ 
are isomorphic. According to Proposition~\ref{prop:partial_atomaton},
this means that the partial atoms of $\cN$ are the atoms of $L$. 
\end{proof}

\begin{example}
\label{ex:atomicity}
All three possibilities for the atomic nature of $\cN$ and $\cN^\rev$ exist:
$\cN_{a}$ of Table~\ref{tab:Na} and its reverse are not atomic.
$\cN_{b}$ of Table~\ref{tab:Nb} is atomic, but its reverse is not.
$\cN_{c}$ of Table~\ref{tab:Nc} and its reverse are both atomic.
Note that all three of these NFAs accept $\Sig^*ab\Sig^*$, where $\Sig=\{a,b\}$.
\qedb
\begin{table}[b]
\begin{minipage}[b]{0.3\linewidth}
\caption{$\cN_a$.}
\label{tab:Na}
{\footnotesize
\begin{center}
$
\begin{array}{|c| c||c| c| }    
\hline
& \ \  \ \ 
&\ \ a \ \ &\ \ b \ \   \\
\hline  
\rightarrow & 0
& \ \{0,1\} \ & \ \{0\} \   \\
\hline  
& 1
&  & \{2\}    \\
\hline  
\leftarrow& 2
&  \{2\} &  \{2\}   \\
\hline  
\end{array}
$
\end{center}}
\end{minipage}
\hspace{0.3cm}
\begin{minipage}[b]{0.3\linewidth}
\caption{$\cN_b$.}
\label{tab:Nb}
{\footnotesize
\begin{center}
$
\begin{array}{|c| c||c| c| }    
\hline
& \ \  \ \ 
&\ \ a \ \ &\ \ b \ \ \\
\hline  
\rightarrow & 0
& \ \{1\} \ & \ \{0\} \ \\
\hline  
 & 1
&  \{1 \} & \  \{1,2 \} \ \\
\hline  
\leftarrow & \ 2 \
 & \ \{1,2\} \ & \ \{ 0\} \ \\
\hline  
\end{array}
$
\end{center}}
\end{minipage}
\hspace{0.55cm}
\begin{minipage}[b]{0.3\linewidth}
\caption{$\cN_c$.}
\label{tab:Nc}
{\footnotesize
\begin{center}
$
\begin{array}{|c| c||c| c| }    
\hline
& \ \  \ \ 
&\ \ a \ \ &\ \ b \ \ \\
\hline  
\rightarrow & 0
& \ \{1\} \ & \ \{0\} \ \\
\hline  
 & 1
&  \{1\} & \{1,2\} \\
\hline  
\leftarrow & \ 2 \
 & \ \{2\} \ & \  \ \\
\hline  
\end{array}
$
\end{center}}
\end{minipage}

\end{table}

\end{example}

\subsection{Extension of Brzozowski's Theorem on DFA Minimization}
\label{sec:extension}

Theorem~\ref{thm:Brz} is the basis for Brzozowski's ``double-reversal'' 
minimization algorithm~\cite{Brz63}:
Given any DFA $\cD$, reverse it to get $\cD^{\rev}$, determinize 
$\cD^{\rev}$ to get $\cD^{\rev\deter}$, reverse $\cD^{\rev\deter}$ to get 
$\cD^{\rev\deter\rev}$, and then determinize $\cD^{\rev\deter\rev}$ to get 
$\cD^{\rev\deter\rev\deter}$. This last DFA is guaranteed to be minimal 
by Theorem~\ref{thm:Brz}, since $\cD^{\rev\deter}$ is deterministic and
$\cD^{\rev\deter\rev}$ has no empty states.
Hence $\cD^{\rev\deter\rev\deter}$ is the minimal DFA equivalent to $\cD$.

Since this conceptually very simple algorithm carries out two 
determinizations, its complexity is exponential in the number of states 
of the original automaton in the worst case. But its 
performance is good in practice, often better than Hopcroft's 
algorithm \cite{TaV05,Wat95}.
Furthermore, this algorithm applied to an NFA still yields an equivalent 
minimal DFA; see~\cite{Wat95}, for example.

As a consequence  of Corollary~\ref{cor:atomic}, we can now generalize 
Theorem~\ref{thm:Brz}:

\begin{theorem}
\label{thm:extension}
For any NFA $\cN$,  $\cN^\deter$ is minimal if and only if $\cN^\rev$ is atomic. 
\end{theorem}

\begin{corollary}
If $\cD$ is a non-minimal  DFA, then $\cD^\rev$ is not atomic. 
\end{corollary}

\section{Reduced Atomic NFAs of a Given Regular Language}
\label{sec:reduced}

The following properties of reduced atomic NFAs were proved in~\cite{BrTa13}. 
A similar approach was used more informally by Sengoku~\cite{Sen92}.

\begin{theorem}[Legality]
\label{thm:unions}
Suppose $L$ is a regular language,  its \'atomaton is
$\cA=({\mathbf A},\Sig,\alpha, {\mathbf I}_A,\{{\mathbf A}_{p-1}\})$, and
$\cB=({\mathbf B},\Sig,\beta,{\mathbf I}_B,{\mathbf F}_B)$ is a trim NFA, where 
${\mathbf B}=\{{\mathbf B}_1,\ldots,{\mathbf B}_r\}$ is a collection of 
sets of positive atom symbols and ${\mathbf I}_B,{\mathbf F}_B\subseteq{\mathbf B}$.
If ${\mathbf B}'\subseteq{\mathbf B}$, define 
$U({\mathbf B}')=\bigcup_{{\mathbf B}_i\in {\mathbf B}'} {\mathbf B}_i$ to be the set of atom symbols 
appearing in the sets ${\mathbf B}_i$  of ${\mathbf B}'$. 
Then $\cB$ is a reduced atomic NFA of $L$ if and only if it satisfies the following
conditions:
\be
\item
\label{cond:in}
$U({\mathbf I}_B)={\mathbf I}_A$.
\item
\label{cond:trans}
For all ${\mathbf B}_i\in {\mathbf B}$, $U(\beta({\mathbf B}_i,a))=\alpha({\mathbf B}_i,a)$.
\item
\label{cond:out}
For all ${\mathbf B}_i\in {\mathbf B}$, we have ${\mathbf B}_i\in {\mathbf F}_B$ 
if and only if ${\mathbf A}_{p-1}\in {\mathbf B}_i$.
\ee
\end{theorem}

\subsection{Enumerating Reduced Atomic NFAs}
If we allow equivalent states, there is an infinite number of atomic NFAs 
accepting a given regular language, but their behaviours are not all distinct.
Hence we consider only reduced atomic NFAs.
The number of trim reduced atomic NFAs can be very large. 
There can be such NFAs with as many as $2^p-1$ non-empty states, 
since there are that many non-empty sets of positive atoms. 

From now on, we drop the curly brackets and commas when representing 
sets of states or sets of atoms in tables.
For example, $\{012,01\}$ stands for $\{ \{0,1,2\},\{0,1\} \}$, and 
$\{A,AB,AC\}$ is used instead of $\{ \{A\}, \{A,B\},\{A,C\} \}$.

\begin{example}
\label{ex:reducedatomic}
The DFA of Table~\ref{tab:dkw} was used in~\cite{KaWe70}. 
It accepts the language $L=\Sig^*(b\cup aa) \cup a$, where $\Sig=\{a,b\}$. 
The quotients of $L$ are
$K_0=\eps^{-1}L=L$, 
$K_1=a^{-1}L=\Sig^*(b\cup aa) \cup a \cup \eps$, and
$K_2=b^{-1}L=\Sig^*(b\cup aa) \cup \eps$.
NFA $\cD^{\rev\deter\rev\trim}$ and the isomorphic trim \'atomaton $\cA^\trim$ with states renamed  are shown in Tables~\ref{tab:drdrkw} and~\ref{tab:akw}.
The positive atoms are
$A=\Sig^*(b\cup aa)$, $B=a$ and $C=\eps$, and
$K_0=A\cup B$, 
$K_1=A\cup B\cup C$,
and $K_2=A \cup C$.

\begin{table}[b]
\begin{minipage}[b]{0.25\linewidth}
\caption{$\cD$.}
\label{tab:dkw}
{\footnotesize
\begin{center}
$
\begin{array}{|c| c||c| c| }    
\hline
& 
&  a  & b  \\
\hline  
\rightarrow & 0
& 1 & 2 \\
\hline  
\leftarrow& 1
&  1  & 2 \\
\hline  
\leftarrow & 2
 &  0 &  2  \\
\hline  
\end{array}
$
\end{center}}
\end{minipage}
\hspace{0.2cm}
\begin{minipage}[b]{0.3\linewidth}
\caption{$\cD^{\rev\deter\rev\trim}$.}
\label{tab:drdrkw}
{\footnotesize
\begin{center}
$
\begin{array}{|c| c||c| c| }    
\hline
&  
& a & b  \\
\hline  
\leftarrow & 12
& &  \\
\hline  
\rightarrow & 01
&  12 &   \\
\hline  
\rightarrow & 012 
&  012,01  & 012,12   \\
\hline  
\end{array}
$
\end{center}}
\end{minipage}
\hspace{1cm}
\begin{minipage}[b]{0.25\linewidth}
\caption{$\cA^ \trim$.}
\label{tab:akw}
{\footnotesize
\begin{center}
$
\begin{array}{|c| c||c| c| }    
\hline
&
& a & b    \\
\hline  
\leftarrow & C
& &  \\
\hline  
\rightarrow & B
&  C &   \\
\hline  
\rightarrow & A 
&  AB  & AC  \\
\hline  
\end{array}
$
\end{center}}
\end{minipage}
\end{table}

\begin{table}[hbt]
\begin{minipage}[b]{0.45\linewidth}
\caption{NFA $\cB_1$.}
\label{tab:b1}
{\footnotesize
\begin{center}
$
\begin{array}{|c| c||c| c| c|c|}    
\hline
& \ \  \ \ 
&\ \ a \ \ &\ \ b \ \   \\
\hline  
\rightarrow & \ AB \
&  \ AB,AC \  & \ AC \    \\
\hline  
\leftarrow& AC
&  \  AB \  &  AC \\
\hline  
\end{array}
$
\end{center}}
\end{minipage}
\hspace{0.4cm}
\begin{minipage}[b]{0.45\linewidth}
\caption{NFA $\cB_2$.}
\label{tab:b2}
{\footnotesize
\begin{center}
$
\begin{array}{|c| c||c| c| }    
\hline
& \ \  \ \ 
&\ \ a \ \ &\ \ b \ \ \\
\hline  
\rightarrow & AB
& \ AB,C \ & \ AC \ \\
\hline  
\leftarrow& \ C \
&  & \\
\hline  
\leftarrow & AC
 &  AB &  AC  \\
\hline  
\end{array}
$
\end{center}}
\end{minipage}
\end{table}

\begin{table}[t]
\begin{minipage}[b]{0.45\linewidth}
\caption{A 5-state NFA.}
\label{tab:b3}
{\footnotesize
\begin{center}
$
\begin{array}{|c| c||c| c| }    
\hline
& \ \  \ \ 
&\ \ a \ \ &\ \ b \ \ \\
\hline  
\rightarrow & A
& \ A,B \ & \ AC \ \\
\hline  
\rightarrow& \ B \
& C & \\
\hline  
\leftarrow & AC
 &  AB &  AC  \\
\hline  
\leftarrow& \ C \
&  & \\
 \hline  
 & AB
 & \ AB,C \ &  A,C  \\
\hline  
\end{array}
$
\end{center}}
\end{minipage}
\hspace{.4cm}
\begin{minipage}[b]{0.45\linewidth}
\caption{A 7-state NFA.}
\label{tab:b4}
{\footnotesize
\begin{center}
$
\begin{array}{|c| c||c| c| c|c|}    
\hline
& \ \  \ \ 
&\ \ a \ \ &\ \ b \ \   \\
\hline  
\rightarrow & \ A \
&  \ A,B \  &  AC     \\
\hline  
\rightarrow& B
&  C &  \\
\hline  
\leftarrow& AC
&  \  AB \  &  AC \\
\hline  
\leftarrow& C
&    &   \\
\hline  
\rightarrow& AB
&   ABC, BC  &  AC \\
\hline  
\leftarrow& ABC
&    ABC,BC   &  AC \\
\hline  
\leftarrow& BC
&  \  C \  &   \\
\hline  
\end{array}
$
\end{center}}
\end{minipage}
\end{table}

Since $L$ is not of the form $L=K^*$, where $K\subseteq \Sig^*$, 
no 1-state NFA exists for $L$.
\be
\item 
For the initial state we could pick state $\{A,B\}$ with two atoms.  From there, the \'atomaton reaches 
$\{A,B,C\}$ under $a$, and  $\{A,C\}$ under $b$. 
        \be
        \item
If we pick $\{A,C\}$
as the second state,  we can cover $\{A,B,C\}$ by $\{A,B\}$ and 
$\{A,C\}$, as  in Table~\ref{tab:b1}. This minimal
atomic NFA turns out to be unique; it is also minimal among all NFAs.
        \item
        We can use $\{A,B,C\}$ as a state and $\{A,C\}$
        for the transition under $b$. This gives an NFA  isomorphic to the DFA of Table~\ref{tab:dkw}.
        \item
        We can use state $\{C\}$
        as shown in Table~\ref{tab:b2}.
        \ee
\item
We can pick two initial states, $\{A\}$ and $\{B\}$. 
        \be
        \item
        If we add $\{C\}$, this leads to the  \'atomaton of Table~\ref{tab:akw}.
        \item
        A 5-state solution is shown in Table~\ref{tab:b3}.
        \ee
\item
We can use three initial states, $\{A\}$, $\{B\}$ and $\{A,B\}$. 
        A 7-state NFA is shown in  Table~\ref{tab:b4}. This 
           is a largest possible reduced solution.\qedb
\ee

\end{example}

The number of minimal atomic NFAs can also be very large. 
\begin{example}
\label{ex:atomicminimal}
Let $\Sig=\{a,b\}$ and consider the language $L=\Sig^*a\Sig^*b\Sig^*=\Sig^*ab\Sig^*$.
The quotients of $L$ are $K_0=L$, $K_1=L\cup b\Sig^*$ and $K_2=\Sig^*$.
The quotient DFA of $L$ is shown in Table~\ref{tab:d}, and its \'atomaton, in Tables~\ref{tab:a} and~\ref{tab:a_relabel} (where the atoms have been relabelled). 
The atoms  are $A=L$, $B=b^*ba^*$ and $C=a^*$, and there is no negative atom.
Thus the quotients are $K_0=L=A$, $K_1=A\cup B$, and $K_2=A\cup B\cup C$.

We find all the minimal atomic NFAs of $L$.
Obviously, there is no 1-state solution.
The states of any atomic NFA are sets of atoms, and 
there are seven non-empty sets of atoms to choose from. 
Since there is only one initial atom, there is no choice: we must take $\{A\}$.
For the transition $(A,a,\{A,B\})$, we can add $\{B\}$ or $\{A,B\}$. 
If there are only two states, atom $C$ cannot be reached. So there is no  2-state atomic NFA.
The results for 3-state atomic NFAs  are summarized in Proposition~\ref{prop:281}. 
\qedb
\end{example}

\begin{table}[hbt]
\begin{minipage}[b]{0.2\linewidth}
\caption{DFA $\cD$.}
\label{tab:d}
{\footnotesize
\begin{center}
$
\begin{array}{|c|c|| c|c|}    
\hline
 & &   a 
&  b   \\
\hline
\hline
\rightarrow & 0 & 1
& 0  \\
\hline  
 & 1 & 
1 &    2  \\
\hline  
\leftarrow &  2 &  2 
&  2 \\
\hline  
\end{array}
$
\end{center}}
\end{minipage}
\hspace{1cm}
\begin{minipage}[b]{0.32\linewidth}
\caption{\'Atomaton $\cA$.}
\label{tab:a}
{\footnotesize
\begin{center}
$
\begin{array}{|c| c||c| c| }    
\hline
&
&a  & b    \\
\hline  
\leftarrow & 2
& 2 &  \\
\hline  
 & 12
&   & 12,2  \\
\hline  
\rightarrow &  012 
&  012,12  &  012   \\
\hline  
\end{array}
$
\end{center}}
\end{minipage}
\hspace{0.35cm}
\begin{minipage}[b]{0.35\linewidth}
\caption{$\cA$ relabelled.}
\label{tab:a_relabel}
{\footnotesize
\begin{center}
$
\begin{array}{|c| c||c| c| }    
\hline
& 
& a & b   \\
\hline  
\leftarrow & C
& C &  \\
\hline  
 & B
&   &   BC  \\
\hline  
\rightarrow &  A 
&  AB  &  A   \\
\hline  
\end{array}
$
\end{center}}
\end{minipage}
\end{table}

\begin{proposition} 
\label{prop:281}
The language $\Sig^*ab\Sig^*$ has 281 minimal atomic NFAs. 
\end{proposition}
\begin{proof}

The only initial state of the \'atomaton $\cA$ corresponds to atom $A$, 
so $\{A\}$ must be included.
To implement the transition
$(A,a,\{A,B\})$ from $\cA$,
either $\{B\}$ or $\{A,B\}$ must be chosen. 
\be
\item
If $\{B\}$ is chosen, then there must be a set containing $C$ but not $A$; otherwise 
the transition 
 $(B,b,\{B,C\})$ cannot be realized.
        \be
        \item
        If $\{B,C\}$ is taken, then $\{C\}$ must be taken, and this makes four states.
        \item
         Hence $\{C\}$ must be chosen, yielding the \'atomaton $\cA=\cN_1$.
         \ee
\item
If $\{A,B\}$ is chosen, then we could choose $\{C\}$, $\{A,C\}$ or $\{A,B,C\}$, 
since $\{B,C\}$ would also require $\{C\}$. Thus there are three cases:
        \be
        \item
        $\{\{A\},\{A,B\},\{C\}\}$ yields $\cN_2$ of Table~\ref{tab:fn1}, if the minimal number of 
        transitions is used. 
        The following transitions can also be added: 
        $(\{A\},a,\{A\})$, $(\{A,B\},a,\{A\})$, $(\{A,B\},b,\{A\})$.
        Since these can be added independently, we have eight more NFAs. 
        Using the maximal number of transitions, we get $\cN_9$ of Table~\ref{tab:fn9}.
        \item
        $\{\{A\},\{A,B\},\{A,C\}\}$ results in $\cN_{10}$ with the minimal number 
          of transitions, and $\cN_{25}$ with the maximal one.
        \item
        $\{\{A\},\{A,B\},\{A,B,C\}\}$ results in $\cN_{26}$ (the quotient DFA) 
          with the minimal number of transitions, and  $\cN_{281}$ with the maximal one.
        \ee
\ee

\begin{table}[h]
\begin{minipage}[b]{0.45\linewidth}
\caption{NFA $\cN_2$.}
\label{tab:fn1}
{\footnotesize
\begin{center}
$
\begin{array}{|c| c||c| c| }    
\hline
& \ \  \ \ 
&\ \ a \ \ &\ \ b \ \ \\
\hline  
\rightarrow & A
& \ AB \ & \ A \ \\
\hline  
 & AB
&  AB  & AB,C \\
\hline  
\leftarrow & \ C \
 & \ C \ & \  \ \\
\hline  
\end{array}
$
\end{center}}
\end{minipage}
\hspace{0.1cm}
\begin{minipage}[b]{0.45\linewidth}
\caption{NFA $\cN_9$.}
\label{tab:fn9}
{\footnotesize
\begin{center}
$
\begin{array}{|c| c||c| c| c|c|}    
\hline
& \ \  \ \ 
&\ \ a \ \ &\ \ b \ \   \\
\hline  
\rightarrow & \ A \
&  \ A,AB \  &  \  A  \   \\
\hline  
 & AB
&  \  A,AB \  & \ A,AB,C \ \\
\hline  
\leftarrow & C
&   \ C  \  &      \\
\hline  
\end{array}
$
\end{center}}
\end{minipage}
\end{table}

\begin{table}[hbt]
\begin{minipage}[b]{0.45\linewidth}
\caption{NFA $\cN_{10}$.}
\label{tab:fn10}
{\footnotesize
\begin{center}
$
\begin{array}{|c| c||c| c| }    
\hline
& \ \  \ \ 
&\ \ a \ \ &\ \ b \ \ \\
\hline  
\rightarrow & A
& \ AB \ & \ A \ \\
\hline  
 & AB
&  AB  & \  AB,AC \ \\
\hline  
\leftarrow & \ AC \
 & \ AB,AC \ & \ A \ \\
\hline  
\end{array}
$
\end{center}}
\end{minipage}
\hspace{0.2cm}
\begin{minipage}[b]{0.45\linewidth}
\caption{NFA $\cN_{25}$.}
\label{tab:fn25}
{\footnotesize
\begin{center}
$
\begin{array}{|c| c||c| c| c|c|}    
\hline
& \ \  \ \ 
&\ \ a \ \ &\ \ b \ \   \\
\hline  
\rightarrow & \ A \
&  \ A,AB \  &  \  A  \   \\
\hline  
 & AB
&  \  A,AB \  & \ A,AB,AC \ \\
\hline  
\leftarrow & AC
&   \ A,AB,AC  \  &   A   \\
\hline  
\end{array}
$
\end{center}}
\end{minipage}
\end{table}

\begin{table}[h]
\begin{minipage}[b]{0.45\linewidth}
\caption{NFA $\cN_{26}$.}
\label{tab:fn26}
{\footnotesize
\begin{center}
$
\begin{array}{|c| c||c| c| }    
\hline
& \ \  \ \ 
&\ \ a \ \ &\ \ b \ \ \\
\hline  
\rightarrow &A
& \ AB \ & \ A \ \\
\hline  
 & AB
&  AB  & \  ABC \ \\
\hline  
\leftarrow & \ ABC \
 & \ ABC \ & \ ABC \ \\
\hline  
\end{array}
$
\end{center}}
\end{minipage}
\hspace{0.2cm}
\begin{minipage}[b]{0.45\linewidth}
\caption{NFA $\cN_{281}$.}
\label{tab:fn281}
{\footnotesize
\begin{center}
$
\begin{array}{|c| c||c| c| c|c|}    
\hline
& \ \  \ \ 
&\ \ a \ \ &\ \ b \ \   \\
\hline  
\rightarrow & \ A \
&  \ A,AB \  &  \ A  \   \\
\hline  
 & AB
&  \  A,AB \  & \ A,AB,ABC \ \\
\hline  
\leftarrow & ABC
&   \ A,AB,ABC \  &   A,AB,ABC   \\
\hline  
\end{array}
$
\end{center}}
\end{minipage}
\end{table}

Also $L$ has 3-state non-atomic NFAs.
The determinized version of NFA $\cN_{10}$ of Table~\ref{tab:fn10} is not minimal.
By Theorem~\ref{thm:extension}, $\cN_{10}^\rev$ is not atomic. But $L^R=\Sig^*ba\Sig^*$;
hence we obtain a non-atomic 3-state NFA for $L$ by reversing $\cN_{10}$ and interchanging $a$ and $b$. There are other non-atomic 3-state solutions.
\end{proof}

One can verify that there is no NFA with fewer than 3 states which
accepts the language $L=\Sig^*ab\Sig^*$.
This implies that every minimal atomic NFA of $L$ is also 
a minimal NFA of $L$.
However, this is not the case with all regular languages, as we will see next.

\subsection{Atomic Minimal NFAs}
 
Recall that Sengoku 
defines an NFA $\cN$ to be in \emph{standard form}~\cite{Sen92}(p.~19) 
if $\cN^{\rev\deter}$ is minimal.
By our Corollary~\ref{cor:atomic}, such an $\cN$ is atomic.
Sengoku makes the following claim~\cite{Sen92}(p.~20):
\begin{quote}
\vskip-0.1cm
\emph{We can transform the nondeterministic automaton into its standard form 
by adding some extra transitions to the automaton. Therefore the number of 
states is unchangeable.}
\end{quote}
\vskip-0.1cm
This claim amounts to stating that any NFA can be transformed to an equivalent  
atomic NFA by adding some transitions. Unfortunately, it is false:
\begin{theorem}
\label{thm:Sengoku}
There exists a language for which no minimal NFA is atomic.
\end{theorem}
\begin{proof}
\vskip-0.1cm
The regular language $L_1$ accepted by DFA $\cD$ of Table~\ref{tab:d_mp}  
is the same as that of an NFA considered in~\cite{MaPo95}(p.~80, Sect.~3). 
NFA $\cD^{\rev\deter\rev}$ and its isomorphic \'atomaton $\cA$ 
with relabelled states are in Tables~\ref{tab:drdr_mp} and~\ref{tab:a_mp}, 
respectively (there is no negative atom). 

Recall that a ``fooling set'' for a regular language $L$ is a set $S=\{(x_i,y_i)\mid x_i,y_i\in\Sig^*, i=1,2,\dots,k \}$ such that $x_iy_i\in L$ for all $i$, and either $x_iy_j\not\in L$ or $x_jy_i\not\in L$ for all $i\neq j$. It is known that every NFA of $L$ needs at least $k$ states, if it has a fooling set of cardinality $k$~\cite{Bir92}.
One verifies that $\{(\eps,b), (a,bb), (aa,bbb), (b,\eps)\}$ is a fooling set for $L_1$. 
Hence every NFA for $L_1$ requires at least four states.

A minimal NFA $\cN_{min}$ of $L_1$ having four states is shown in 
Table~\ref{tab:n_mp}; it is not atomic and it is not unique. 
We try to construct a 4-state atomic NFA $\cN_{atom}$ equivalent to $\cD$. 
\begin{table}[hbt]
\begin{minipage}[b]{0.19\linewidth}
\caption{$\cD$.}
\label{tab:d_mp}
{\footnotesize
\begin{center}
$
\begin{array}{|c|c|| c|c|}    
\hline
 & &  a 
&  b   \\
\hline
\hline
\rightarrow & 0 
& 1 & 2  \\
\hline  
 & 1 
 &  3 &    4  \\
\hline  
\leftarrow & \ 2 \ 
&  5 &  4  \\
\hline  
 & 3 & 
3 &    1  \\
\hline  
 & 4 & 
6 &    2  \\
\hline  
\leftarrow & 5 & 
7 &    2  \\
\hline  
 & 6 & 
3 &    8  \\
\hline  
\leftarrow & 7 & 
7 &    7  \\
\hline  
 & 8 & 
6 &    7  \\
\hline  
\end{array}
$
\end{center}}
\end{minipage}
\hspace{0.03cm}
\begin{minipage}[b]{0.38\linewidth}
\caption{$\cD^{\rev\deter\rev}$.}
\label{tab:drdr_mp}
{\footnotesize
\begin{center}
$
\begin{array}{|c| c||c| c| }    
\hline
& 
& a  & b   \\
\hline  
\leftarrow & 257
& 257,04578 &  \\
\hline  
\rightarrow & 04578
&  12678 &   257  \\
\hline  
 & 12678 
&  & 04578,03-8   \\
\hline  
\rightarrow & 03-8
&   &   12678  \\
\hline  
 & 1-8
 &  03-8 &   \\
\hline  
\rightarrow & 0-8
&  1-8,0-8 &   1-8,0-8  \\
\hline  
\end{array}
$
\end{center}}
\end{minipage}
\hspace{0.83cm}
\begin{minipage}[b]{0.32\linewidth}
\caption{ $\cA$.}
\label{tab:a_mp}
{\footnotesize
\begin{center}
$
\begin{array}{|c| c||c| c| }    
\hline
&
& a & b    \\
\hline  
\leftarrow & A
& AB &  \\
\hline  
\rightarrow & B
&  C & A  \\
\hline  
 &  C 
&   & \ BD \  \\
\hline  
\rightarrow &  D 
&   & C \\
\hline  
 &  E 
& D  &     \\
\hline  
\rightarrow &  F 
&  EF  &  EF \\
\hline  
\end{array}
$
\end{center}}
\end{minipage}
\end{table}
First, we note that quotients corresponding to the states of $\cD$ can be expressed 
as sets of atoms as follows:
$K_0=\{B,D,F\}$, $K_1=\{C,E,F\}$, $K_2=\{A,C,E,F\}$, $K_3=\{D,E,F\}$,
$K_4=\{B,D,E,F\}$, $K_5=\{A,B,D,E,F\}$, $K_6=\{C,D,E,F\}$, $K_7=\{A,B,C,D,E,F\}$, and
$K_8=\{B,C,D,E,F\}$. One can verify that these are the states of the determinized 
version of the \'atomaton, which is isomorphic to the original DFA $\cD$. 
Now, every state of $\cN_{atom}$ must be a subset of a set of atoms of some quotient, 
and all these sets of atoms of quotients must be covered by the states of $\cN_{atom}$.
We note that quotients $\{B,D,F\}$, $\{C,E,F\}$, and $\{D,E,F\}$
do not contain any other quotients as subsets, while all the other quotients do.
It is easy to see that there is no combination of three or fewer sets of atoms, 
other than these three sets, that can cover these quotients. 
Since our aim is to find a four-state atomic NFA, and because we also need a set 
containing the atom $A$, we have to use these three sets as states of $\cN_{atom}$. 
To use only one set of atoms with $A$, that set has to be
a subset of every quotient having $A$. So it must be
a subset of $\{A,E,F\}$. If we use $\{A\}$ as a state, then by the transition 
table of the \'atomaton, there must be at least one more state to cover 
$\{A,B\}$. Similarly, if we use $\{A,E\}$, then we must have another state to cover 
$\{A,B,D\}$. If we use $\{A,F\}$, then we must have a state to cover 
$\{A,B,E,F\}$. And if we use $\{A,E,F\}$, then we must have a state to cover 
$\{E,F\}$. We conclude that a smallest atomic NFA has at least five states.
There is a five-state atomic NFA, as 
shown in Table~\ref{tab:n5_mp}. It is not unique. 

Since there does not exist a four-state atomic NFA equivalent to the DFA $\cD$,
it is not possible to convert the non-atomic 
minimal NFA $\cN_{min}$ to an atomic NFA by adding transitions.
\end{proof}

\begin{table}[t]
\begin{minipage}[b]{0.3\linewidth}
\caption{NFA $\cN_{min}$.}
\label{tab:n_mp}
{\footnotesize
\begin{center}
$
\begin{array}{|c|c|| c|c|}    
\hline
 & & \  a \
& \ b \  \\
\hline
\hline
\rightarrow & 0 
& \ 1 \ & 1,2  \\
\hline  
 & 1 
 & 3 & \  0,3 \ \\
\hline  
\leftarrow & \ 2 \ 
& \ 0,2,3 \ &  \\
\hline  
 & \ 3 \ 
& \ 3 \ & 1 \\
\hline  
\end{array}
$
\end{center}}
\end{minipage}
\hspace{0.5cm}
\begin{minipage}[b]{0.45\linewidth}
\caption{$\cN_{atom}$.}
\label{tab:n5_mp}
{\footnotesize
\begin{center}
$
\begin{array}{|c| c||c| c| }    
\hline
& \ \  \ \ 
&\ \ a \ \ &\ \ b \ \   \\
\hline  
\rightarrow & BDF
& CEF &CEF,AEF  \\
\hline  
 & \ CEF \
& DEF  & \ BDF,DEF \  \\
\hline  
\leftarrow & AEF
&  \ BDF, AEF, DEF \ & EF  \\
\hline  
 & \ DEF \
& DEF  & \ CEF \  \\
\hline  \hline
 & \ EF \
&\ DEF  \ & EF \\
\hline  
\end{array}
$
\end{center}}
\end{minipage}
\vskip-0.3cm
\end{table}

\vskip-0.1cm
In summary, Sengoku's method cannot always find the minimal NFAs, but 
 it is able to find all atomic minimal NFAs.
His minimization algorithm proceeds by 
``merging some states of the normal nondeterministic automaton''.
This is similar to our search for subsets of atoms that satisfy 
Theorem~\ref{thm:unions}.

\section{Conclusions}
\label{sec:conc}
For any NFA $\cN$, we introduced a natural set of languages, the partial atoms, and constructed a new NFA, which we proved to be isomorphic to $\cN^{RDR}$---an NFA studied by Sengoku.
For any regular language $L$, we introduced a natural set of languages, the atoms of $L$; we then constructed an NFA $\cA$, the \'atomaton of $L$, which we proved to be isomorphic to the NFA $\cN^{RDMR}$, also studied by Sengoku.
 We  introduced atomic automata, and generalized 
Brzozowski's method of minimization of DFAs by double reversal.
We studied atomic NFAs associated with a given regular language and,
contrarily to Sengoku's claim, proved that not every language has an atomic 
minimal NFA.

For completeness we mention that the quotient complexity (equivalent to 
state complexity) of atoms of regular languages was studied in~\cite{BrTa13} and~\cite{BrDa13}.

\end{document}